\documentclass[11pt]{article}

\usepackage{fullpage}

\usepackage{amssymb,hyperref,amsthm,mathtools,enumerate,thmtools,thm-restate}
\usepackage[nameinlink]{cleveref}
\usepackage[vlined,boxed,commentsnumbered, ruled,linesnumbered]{algorithm2e}
\SetKwInOut{Given}{Input}
\SetKwFor{RepTimes}{repeat}{times}{}

\usepackage{dsfont}

\usepackage[noend]{algpseudocode}

\usepackage[margin=1in]{geometry}

\declaretheorem[name=Theorem, parent=section]{theorem}
\declaretheorem[name=Corollary, sibling=theorem]{corollary}
\declaretheorem[name=Proposition, sibling=theorem]{prop}
\declaretheorem[name=Lemma, sibling=theorem]{lemma}
\declaretheorem[name=Claim, sibling=theorem]{claim}

\newtheorem*{remark}{Remark}

\theoremstyle{definition}
\declaretheorem[name=Definition, sibling=theorem]{definition}

\theoremstyle{remark}
\declaretheorem[name=Observation, sibling=theorem]{observation}

\newcommand\E{\mathbb{E}}

\newcommand\R{\mathbb{R}}

\newcommand\N{\mathbb{N}}

\newcommand\F{\mathbb{F}}

\newcommand\card[1]{\left| {#1} \right|}

\newcommand\set[1]{{\left\{ #1 \right\}}}

\newcommand\Prob[2]{{\Pr_{#1}\left[ {#2} \right]}}

\newcommand\ceil[1]{\left\lceil{#1}\right\rceil}

\newcommand\floor[1]{\left\lfloor{#1}\right\rfloor}

\newcommand\ip[2]{\left\langle {#1}, {#2} \right\rangle}

\newcommand\cube[1]{ { \set{0,1}^{#1} } }

\newcommand\eps{\epsilon}
\newcommand\poly{\text{poly}}

\newcommand\ext{{\sf Ext}_d}

\newcommand*\xor{\mathbin{\oplus}}

\newcommand*\bxor{\mathbin{\bigoplus}}

 \newcommand{\Int}{\textsf{Int}}

\newcommand\parentheses[1]{{\left({#1}\right)}}
\newcommand\fc[2]{{\widehat{#1}\parentheses{#2}}}

\newcommand\ind[1]{{\mathds 1}_{#1}}

\newcommand{\cD}{{\cal D}}

\newcommand{\cF}{{\cal F}}

\newcommand{\cO}{{\cal O}}
\newcommand{\cP}{{\cal P}}
\newcommand{\cS}{{\cal S}}
\newcommand{\cT}{{\cal T}}
\newcommand{\cW}{{\cal W}}

\newcommand{\oracle}[1]{\mathcal{O}_{#1}}
\newcommand{\blr}[1]{\ensuremath{\text{\sc XorTest}_{#1}}}

\newcommand{\dist}{{\text{dist}}}
\newcommand{\Dist}{{\text{Dist}}}

\newcommand{\short}[2]{#2}

\title{Property Testing with Online Adversaries}

\author{{Omri Ben-Eliezer \thanks{Massachusetts Institute of Technology, Cambridge, MA, USA. Email: \texttt{omrib@mit.edu}}} 
\and { Esty Kelman \thanks{Boston University, Boston, MA, USA, and Massachusetts Institute of Technology, Cambridge, MA, USA. Email: \texttt{ekelman@mit.edu}, supported by the Computer Science Department, Boston University.} } 
\and {Uri Meir \thanks{Tel-Aviv University, Tel-Aviv, Israel. Email: \texttt{urimeir.cs@gmail.com}}} \and {Sofya Raskhodnikova\thanks{Boston University, Boston, MA, USA. Email: \texttt{sofya@bu.edu}}}}

\date{\today}

\begin{document}

\maketitle

\begin{abstract} 
  The online manipulation-resilient testing model, proposed by Kalemaj, Raskhodnikova and Varma (ITCS 2022 and Theory of Computing 2023), studies property testing in situations where access to the input degrades continuously and adversarially.
    Specifically, after each query made by the tester is answered, the adversary can intervene and either erase or corrupt $t$ data points. In this work, we investigate a more nuanced version of the online model in order to overcome old and new impossibility results for the original model. We start by presenting an optimal tester for linearity and a lower bound for low-degree testing of Boolean functions in the original model. We overcome the lower bound by allowing batch queries, where the tester gets a group of queries answered between manipulations of the data. 
    Our batch size is small enough so that function values for a single batch on their own give no
    information about whether the function is of low degree.
    Finally, to overcome the impossibility results of Kalemaj et al.\ for sortedness and the Lipschitz property of sequences, we extend the model to include $t<1$, i.e., adversaries that make less than one erasure per query. 
    For sortedness, we characterize the rate of erasures for which online testing  can be performed, exhibiting
    a sharp transition from optimal query complexity to impossibility of testability (with any number of queries).
    Our online tester works for a general class of local properties of sequences. 
    One feature of our results is that we get new (and in some cases, simpler) optimal algorithms for several properties in the standard property testing model.

\end{abstract}

\thispagestyle{empty} 
\newpage
\tableofcontents
\thispagestyle{empty} 
\newpage
\setcounter{page}{1}

\section{Introduction}

The online manipulation-resilient testing model, proposed by Kalemaj, Raskhodnikova and Varma \cite{KalemajRV22}, studies property testing in contexts where access to the input  is controlled by an adversary and degrades over time. Their motivation includes situations where access to data is restricted because of privacy or is misrepresented by someone trying to cover fraud while having to release some data in response to subpoenas.
Another motivation is testing properties of a massive ever-changing object, where 
probing the object affects it. For example, a navigation app might adjust its estimate of which routes are congested in response to queries. Such situations are plentiful in data analysis of modern social and transportation networks. 
 
Modeling access to data as degrading adversarially allows the algorithm designer to ensure that algorithms work in situations where the degradation to access is too complicated to model accurately or where it is desirable to avoid relying on distributional assumptions in the algorithm analyses. 
The general goal is to understand what properties of the input can be gleaned without any distributional assumptions, and even in extremely adversarial regimes, where access to the data or the data itself is manipulated by an adversary in response to the algorithm's actions during its execution. From the theoretical point of view, this investigation sheds light on the structure of witnesses for different properties.

In the online manipulation-resilient testing model of~\cite{KalemajRV22}, after each query to the input object is answered, the adversary can manipulate (i.e., erase or corrupt) a fixed number of input values. The input is represented by a function $f$ on an arbitrary finite domain. The algorithm accesses it by specifying a query point $x$ in the domain and receiving the answer to the query from the oracle that represents the current state of the manipulated input. In addition to allowing us to design fast testers that overcome old and new lower bounds, studying batch queries is motivated by a connection to Maker-Breaker games, discussed later (in \Cref{sec:intro-maker-breaker-games}).

In this work, we investigate a more nuanced version of the online manipulation-resilient testing model. We allow the adversary to either make changes at a fixed rate (as in the model of~\cite{KalemajRV22}) or accumulate the number of allowed manipulations in order to utilize them at any subsequent step. We also study arbitrary rates of erasures in addition to integer rates investigated by~\cite{KalemajRV22}, e.g., allowing the adversary to manipulate one point after every other query. Finally, we also study batch queries, where the tester gets a group of queries answered between manipulations of the data.

We give online testers for several well studied properties. 
For Boolean functions on the domain $\{0,1\}^n$,
we investigate linearity and, more generally, the class of functions of low degree,
that is, of degree at most $d$ over $\F_2$
for specified $d$. (Note that low-degree testing is equivalent to testing Reed-Muller codes.) For linearity, our online tester has optimal query complexity. For the degree-$d$ property, we present a lower bound that nearly matches the upper bound in the concurrent work of Minzer and Zheng~\cite{MinZ} and then show how to overcome the lower bound by using batch queries. 
Finally, we consider local properties of functions $f:[n]\to\R$. 
(We use $[n]$ to represent $\set{1,2,\dots,n}$). 
Kalemaj et al.\ proved that two important properties in that class---sortedness and the Lipschitz property---are not testable with any number of queries when the adversary erases one point per query. To overcome this impossibility result, we consider adversaries that make less than one erasure per query. For sortedness, we characterize the rate of erasures for which online testing can be performed, exhibiting
    a sharp transition from optimal query complexity to impossibility of testability. Our online tester works for general local properties. Moreover, we show how to obtain optimal erasure rate with batches of size 2.

Our online testers avoid querying manipulated data by heavily relying on randomness: the marginal distribution of each query is nearly uniform on a large subset of the domain. 
This ensures that each query is unlikely to be a point previously manipulated by the adversary (for any adversarial strategy). 
Consequently, we can separately analyze the probabilities of two important events: selecting a set of queries that exhibit a violation of the property (i.e., a \emph{witness}) and querying data modified by the adversary. 
When selecting a witness is significantly more likely than encountering a modification, the tester is likely to find a witness that has not been tampered with. This approach allows us to design online testers that are resilient to both erasures and corruptions.

\subsection{Our Results and Techniques}\label{sec:results-low-degree} 
We study two types of functions: Boolean functions $f: \{0,1\}^n\to \{0,1\}$ and real-valued functions $f:[n]\to \R$, i.e., sequences. The testing algorithm gets oracle access to the input function via one of the specified online adversarial oracles. Initially, the oracle is identical to the input function. The oracles (formally defined in Section~\ref{sec:model-definitions}) vary in how they are allowed to modify the input. The main distinctions are erasures  vs.\ arbitrary corruptions (of input values) and when modifications can occur. In the model of~\cite{KalemajRV22}, the adversary gets a parameter $t\in\N$ and is allowed to modify $t$ function values after answering each query. We call such an adversary \emph{$t$-online fixed-rate}. We extend the definition to $t\in \R^+$ and also consider \emph{$t$-online budget-managing} adversaries that get allocated $t$ modifications after answering each query, but can use their allocation at any subsequent time step in the computation. (See \Cref{def:fixed-rate_budget-managing}). Budget-managing adversaries are at least as powerful as fixed-rate adversaries. But, for some of the properties we study, we are able to match hardness results for (weaker) fixed-rate adversaries with algorithms that work even in the presence of (stronger) budget-managing adversaries.

In addition to the corruption rate $t$, our algorithms get the proximity parameter $\eps\in (0,1),$ as in the standard property testing model, and have to distinguish functions that have the specified property from functions that differ in at least an $\eps$ fraction of the domain from any function with the property. See \Cref{sec:model-definitions} and \Cref{def:online_tester} for formal definitions.

We start our investigation with Boolean functions on the Boolean cube. The properties we study are linearity and, more generally, of being a polynomial of degree at most $d$ for a given $d\in\N.$ 

\subsubsection{Linearity Testing}\label{sec:intro-linearity}
Introduced in the pioneering work of \cite{BLR93}, linearity testing has been extensively investigated; see,  e.g., \cite{BellareGLR93, BellareS94, FeigeGLSS96, BellareCHKS96, BellareGS98, Trevisan98, SudanT98,  SamorodnitskyT00, HastadW03,Ben-SassonSVW03, Samorodnitsky07, SamorodnitskyT09, ShpilkaW06, KaufmanLX10, KalemajRV22} and the survey in~\cite{RasR16}. 
 A function $f:\cube{n}\to\{0,1\}$ is called \emph{linear} if  $f(x_1) + f(x_2) = f(x_1\xor x_2)$ 
 for all $x_1,x_2\in\cube{n},$ where addition is mod 2, and $\xor$ denotes bitwise XOR.
We present an optimal tester for linearity that is resilient to online erasures, as well as online corruptions, even when the adversary is budget-managing.
  
\begin{restatable}[Optimal online linearity tester]{theorem}{OptimalOnlineLinearityTesterThm}
    \label{thm:linearity_tester}
   There exists a constant $c>0$ such that for all $n\in\N, \eps \in (0, 1/2]$, and $t$ satisfying $t \log^2 t \leq c \cdot\eps^{2.5} 2^{n/2}$, 
   there exists an $\eps$-tester for linearity of functions $f:\{0,1\}^n\to \{0,1\}$ that works in the presence of a $t$-online erasure (or corruption) budget-managing adversary and makes  $O\left(\max\set{1/\eps, \log t}\right)$ queries. In the case of erasure adversary, the tester has 1-sided error.
\end{restatable}

In contrast, the linearity tester from \cite{KalemajRV22} works for a smaller range of $t$ (specifically, $t \leq c_0\cdot\eps^{5/4} 2^{n/4}$ for some constant $c_0$)  and makes $O\big(\min\big( \frac{1}{\eps}\log \frac{t}{\eps}, \frac t \eps\big)\big)$ queries. 

The linearity tester in  \Cref{thm:linearity_tester} has optimal query complexity both for fixed-rate and budget-managing adversaries, in the case of erasures (and, consequently, in the more challenging setting with corruptions). 
This follows from the known query lower bounds of $\Omega\left(1/\eps\right)$ with no erasures and $\Omega(\log t)$ for fixed-rate adversarial erasures \cite[Theorem 1.4]{KalemajRV22}).

Our online linearity tester improves on the tester in \cite{KalemajRV22} both in terms of query complexity and in terms of simplicity of the tester.
The classical linearity test of \cite{BLR93} looks for witnesses of nonlinearity consisting of three points $x_1,x_2$ and $x_1\xor x_2$ that satisfy
$f(x_1) + f(x_2) \neq f(x_1\xor x_2)$. Kalemaj et al.\ generalized it to witnesses consisting of any even number of points and their XOR. Let \blr{k} denote the test that picks $k$ points uniformly at random and checks if these points and their XOR form a witness of nonlinarity.
 We improve upon the soundness analysis of this 
test and use it to get an optimal online-erasure-resilient linearity tester. 
 Specifically, in \Cref{lem:krv_k_test}, we show that $\blr{k}$ rejects every function that is $\eps$-far from linearity with probability proportional to $k \eps$ (as long as $k\eps$ is at most a small constant).
It implies that $\blr{k}$ can be used in the standard linearity testing setting (without erasures) to obtain another optimal $O(1/\eps)$-query tester: $\blr{k}$ can be repeated $O(1/(k\eps))$ times to obtain constant probability of error. To the best of our knowledge, only testers based on the BLR test (i.e., \blr{2}) have been analyzed in the standard linearity testing setting.

Another important ingredient in the analysis of our online linearity tester is bounding the probability of seeing an erasure. A good bound for this event ensures that our tester is resilient to corruptions (not just erasures). It also allows us to obtain a clean online tester that is based on multiple simulations of $\blr{k}$, where each simulation initially samples $2k$ points to fool the adversary and then selects a random subset of size $k$ for the XOR as the final query. In contrast, the online linearity tester of \cite{KalemajRV22} is more complicated: it relies on work investment strategy. 

\subsubsection{Low-Degree Testing}\label{sec:intro-low-degree}
Low-degree testing, equivalent to local testing of Reed-Muller codes, is a natural generalization of linearity testing. It has been investigated, e.g.,
in~\cite{BabaiFL91,BabaiFLS91,GemmellLRSW91,FeigeGLSS96,FriedlS95,RubinfeldS96, RazS97, AlonKKLR05, AroraS03, MoshkovitzR08, Moshkovitz17, KaufmanR06, Samorodnitsky07, SamorodnitskyT09, JutlaPRZ09, BKSSZ10,HaramatySS13, Ron-ZewiS13, DinurG13,kaufman2022improved,KalemajRV22}.
For a $d\in\N$, let $\cP_d$ denote the set of all polynomials $p:\{0,1\}^n\to \{0,1\}$ of degree at most $d$, that is, functions $p(x)$ that can be represented as
a sum of monomials\footnote{To be consistent with previous work, we allow $S=\emptyset$. The property $\cP_1$ is \emph{affinity}, and linear functions (discussed in \Cref{sec:intro-linearity}) are affine with the additional requirement that the constant in the polynomial representation is 0.} of the form $\prod_{i\in S} x[i]$, where~$x=(x[1],\dots,x[n])$ is a vector of~$n$ bits and $|S|\leq d$.
 In the standard property testing model, $\cP_d$ can be $\eps$-tested with $O(1/\eps+2^d)$ queries by repeating the following test of Alon et al.~\cite{AlonKKLR05}: select $d+1$ points in $\{0,1\}^n$ uniformly at random, query all of their linear combinations $f$, and accept iff the sum of the returned values is 0.  This tester was analyzed by Alon et al.~\cite{AlonKKLR05}, with an asymptotic improvement by 
 Bhattacharyya et al.~\cite{BKSSZ10}, and the matching lower bound of $\Omega(1/\eps+2^d)$ on query complexity was proved in \cite{AlonKKLR05}.

As $d$ grows, the test of Alon et al.\ becomes more structured: the number of points queried is exponential in the number of points selected at random. This makes it hard to adapt to the online model, since the adversary can predict and erase the points needed by the tester. We show that there is an underlying reason for this difficulty: low-degree testing for $d>1$ is strictly harder than testing linearity in terms of the dependence on $t$, the rate of erasures.

\begin{restatable}[Lower bound for low-degree testing]{theorem}{LowerBoundForLowDegreeTestingThm}
    \label{thm:degree-d-lb}
    Fix an integer $d > 1$, and let $\cP_d$ be the property of being a polynomial of degree at most $d$ over $\F_2$. 
   There exists $n_0=n_0(d)$, such that for all $n\geq n_0$ and 
   $\eps \in(0, 1/3]$, every $\eps$-tester of functions $f \colon \{0,1\}^n\to \{0,1\}$ for $\cP_d$  that works in the presence of a $t$-online fixed-rate erasure adversary must make $\Omega\left( \log^d t \right)$ queries.   
\end{restatable}

Our lower bound is nearly tight in terms of the dependence on $t$ and $d$: in a concurrent work, Minzer and Zheng~\cite{MinZ} show that $\cP_d$ can be tested with $O\left(\frac 1 \eps \log^{3d+3}\frac t \eps\right)$ queries with $t$-online fixed-rate erasure adversaries. Thus, the query complexity of $\cP_d$ is $\log^{\Theta(d)} t$ (for constant $\eps$).

To prove our lower bound, we define an extended representation $\ext(x)$ for each vector $x \in \F_2^n$, where each entry of $\ext(x)$ is an evaluation of a monomial of degree at most $d$ on input $x$. We consider the adversarial strategy of erasing function values on all points whose extended representations are in the span of the extended representations of previous queries. We use \cite[Lemma 1.4]{Ben-EliezerHL12} (or, equivalently,  \cite[Theorem 1.5]{KeevashS05}) to demonstrate that when the number of queries is, roughly, at most $\binom{\log t}{d}$, the adversary can successfully execute this strategy. Finally, we apply Yao's minimax principle to show that in this case no online tester can distinguish random polynomials of degree $d$ from random functions.

\medskip

We overcome the lower bound by allowing batch queries, where the tester gets a group of $b$ queries answered between manipulations of the data. The original model corresponds to $b=1$.
Since the atomic tester of Alon et al.\ makes $2^{d+1}$ queries, it can be used directly in an online tester with batch size $b=2^{d+1},$ achieving query complexity $O(\frac 1 \eps + 2^d)$, the same as in the offline regime.
For completeness, in \Cref{sec:low-degree-testing-with-large-batches}, we analyze the rate of erasures $t$ achievable in this setting. 
It is natural to ask for which batch sizes we can achieve overhead polynomial in $d\log t$ over the offline query complexity. We show how to do it for $b=2^{d-1}$, i.e., with the batch size equal to one quarter of the points needed for the smallest witness of not satisfying $\cP_d.$
In particular, for quadraticity, it corresponds to batches of size $b=2$, a natural extension of testing linearity with batches of size 1.

\begin{restatable}
{theorem}
    {dDegChainOfCubesThm}\label{thm:d_deg:chain_of_cubes}
    There exists a constant $c>0$ such that for all $n,d\in\N,\ \eps \in (0, 1/2)$, and $t$ satisfying 
    $d <  n /11$ and $t \log^7 t\leq c\cdot \eps^{2.5}d^{-7} 2^{(n-11d)/2}$,
    there exists an $\eps$-tester for property $\cP_d$ of functions $f:\{0,1\}^n\to \{0,1\}$ that works in the presence of a $(2^{d-1},t)$-online erasure budget-managing adversary and makes
    $O\left(1/\eps + 2^{3d} \left(d +\log t\right)^3 \right)$ queries.
\end{restatable}

Our online low-degree tester is a natural generalization of our linearity tester and of the tester of \cite{AlonKKLR05}. The witnesses we consider generalize the $(d+1)$-dimensional cubes formed by the points queried by the tester of \cite{AlonKKLR05} to \emph{chains of cubes} (see \Cref{fig:chain_of_cubes}).

We show chains of cubes are viable witnesses as they are a special case of \emph{$k$-local characterizations} defined in the seminal work of Kaufman and Sudan~\cite{KS08} on algebraic property testing.
In addition, our \emph{chains-of-cubes} test shares important features with the test of Alon et al.
This allows us to show, by generalizing an argument from~\cite{BKSSZ10}, that for small values of $\eps$, the probability of seeing a violation increases linearly in the size of the witness.

To test with a chain of cubes, the tester first declares $d-1$ directions that form a linear subspace~$A$. Then, intuitively, it runs our online linearity tester, replacing each query  $x$ with a batch query $x + A$ (all points in this affine subspace).
To analyze erasure resilience, whenever the adversary erases point $x,$ we allow it to erase the entire space $x+A$ for free.
This allows us to analyze the probability of seeing an erasure over the quotient group $\F_2^n / A$, which is isomorphic to $\F_2^{n-d+1}$, essentially reducing the analysis to that of the online linearity tester (because, over this space, each query and erasure are of a single point).

\subsubsection{Testing Local Properties}\label{sec:intro-local-properties}

Next we investigate properties of sequences, represented by functions $f : [n] \to \mathbb{R}$.
Kalemaj et al.\ showed that two fundamental properties of sequences, \emph{sortedness} and the \emph{Lipschitz property}, are not testable with any number of queries in the presence of a fixed-rate adversary making $t=1$ erasures per query. 
Are these properties testable in models with fewer erasures than queries?

We explore this question in depth. The answer is generally positive, but differs between the two 
adversarial manipulation
models we study: fixed-rate and budget-managing. Our results are optimal, and the query complexity in many regimes matches that of the standard (offline) property testing model. Our online testers work in the general framework of local properties \cite{BE19}.

Before stating our results, we define sortedness, the Lipschitz property, and the general class of local properties.
 A sequence $f : [n] \to \mathbb{R}$ is sorted if $f(x) \leq f(y)$ for all $x < y$, where $x,y\in[n]$.
Sortedness is one of the most investigated properties in the context of property testing (see \cite{EKKRV00,DGLRRS99,Ras99,BGJRW12,CS13,Belovs18,PRV18} and the survey in~\cite{Enc1}). 
A sequence $f : [n] \to \mathbb{R}$ is Lipschitz if $|f(x+1) - f(x)|\leq 1$ for all $x \in [n-1]$. The Lipschitz property has been studied in \cite{JhaR13,CS13, DixitJRT13, BermanRY14, AwasthiJMR16, ChakrabartyDJS17} and has applications to data privacy.
Both sortedness and the Lipschitz property belong to the class of local properties, defined by
 \cite{BE19}. A property $\cP$ of sequences $f \colon [n] \to \R$ is \emph{local}\footnote{This class of properties is called 2-local in \cite{BE19}, where more general $k$-local properties are defined as characterized by families of forbidden $k$-tuples. For simplicity and clarity of exposition, we focus on 2-local properties, but our results easily generalize to $k$-local properties (with increased batch sizes).} if there exists a family $\cF$ of \emph{forbidden} pairs $(a,b)\in \R^2$ such that

\begin{center}
$f\in\mathcal{P}$ $\Leftrightarrow$ $\forall i \in [n-1] \ \forall (a,b) \in \mathcal{F}\colon (f(i), f(i+1)) \neq (a, b)$.
\end{center}
Then we say that $\mathcal{P}$ is characterized by the family $\mathcal{F}$.
 Sortedness is characterized by $\cF=\{(a,b): a > b\}$; that is, a sequence is sorted if and only if it does not contain a pair of \emph{consecutive} elements with decreasing values. The Lipschitz property is characterized by $\cF=\{(a,b):|a-b| > 1\}$.
 
\paragraph{Results for batch size $b=1$.}
We give tight bounds on the rate of erasures and corruptions online algorithms can handle in the presence of \emph{budget-managing} adversaries.
\begin{restatable}[Testing local properties in presence of budget-managing adversary with batch size $1$]{theorem}{TestingLocalPropertiesBudgetManagingTHM}
\label{thm:sortedness-budget-managing}
There exist absolute constants $0 < c < C$ satisfying the following. For all $\eps > 0$ and $n \geq C/\eps$:
\begin{enumerate}
\item For every local property $\mathcal{P}$ of sequences $f \colon [n] \to \R$, and every $t \leq c \eps$, there exists an $\eps$-tester for $\mathcal{P}$ 
that works in the presence of a $t$-online erasure budget-managing adversary and makes $O\big(\frac{\log (\eps n)} \eps\big)$ queries.

\item\label{item:impossibility-sortedness} For every $t \geq C \eps$, there is no $\eps$-tester for sortedness 
of sequences $f \colon [n] \to \R$ that works in the presence of a $t$-online erasure budget-managing adversary (with any number of queries).
\end{enumerate}
Both of these results hold (with the same parameters) if erasures are replaced with corruptions. In the case of erasure adversary, the tester has one-sided error.
\end{restatable}

The positive result in Item 1 of \Cref{thm:sortedness-budget-managing} matches the query complexity of the offline optimal tester from \cite{BE19} while also attaining optimal resilience guarantees.

\Cref{thm:sortedness-budget-managing} highlights a dramatic phase transition 
in the threshold regime where $t = \Theta(\eps)$: when $t = \omega(\eps)$, sortedness is not testable at all; whereas when $t = o(\eps)$, it is testable (and in fact, all local properties are testable) with offline-optimal query complexity!

For fixed-rate adversaries, the threshold rate is arbitrarily close to $1$. The negative result is established in \cite{KalemajRV22}; the positive result is stated in \Cref{prop:local-fixed-rate} and proved in \Cref{sec:local}.

\paragraph{Results for batch size $b=2$.}
For batch size one, we have seen that the threshold rate for sortedness in the presence of a budget-managing adversary is $\Theta(\eps)$, i.e., less than one.  Batches of size two result in a dramatically different picture: we can tolerate as many as $\tilde{\Omega}(n)$ erasures or corruptions between consecutive batches while maintaining optimal query complexity! 

\begin{restatable}{theorem}{SortednessBatchTwoTHM}
\label{thm:sortedness_batch_2}
There exists $c > 0$ satisfying the following. Let $\eps > 0$ and $n \in \N$. 
For every local property $\mathcal{P}$ of sequences $f \colon [n] \to \R$, and every $t \leq \frac{c \eps^2 n}{\log^2 \eps n}$, there exists an $\eps$-tester for $\mathcal{P}$ with batch size $b=2$ that works in the presence of a $t$-online erasure (or corruption) budget-managing adversary and makes $O\big(\frac{\log (\eps n)} \eps\big)$  queries. In the case of erasures, the tester has one-sided error.
\end{restatable}
\paragraph{Technical overview: Pair tester for local properties.} 
Our main technical contribution here shows that a simple and generic pair tester (see Algorithm~\ref{alg:pair_tester_local_properties} and~\Cref{lem:pair_tester_offline}), which queries $f$ on pairs of the form $(x,x+2^i) \in [n]^2$, 
works for
all local properties of sequences in the \emph{offline} property testing model, with query complexity of $O\big( \frac{\log (\eps n)}{\eps} \big)$. This matches known lower bounds on the query complexity of sortedness \cite{ChSe14, Belovs18}.
Our upper bound is the same as that obtained in \cite{BE19}, but the new (pair) tester is simpler and has the additional feature of being resilient to online manipulations. 
Notably, there exist pair testers which obtain this query complexity for specific classes of local properties;
see, e.g., the work of Chakrabarty, Dixit, Jha, and Seshadhri \cite{ChakrabartyDJS17} on bounded derivative properties. Our work generalizes this result to all local properties of sequences. The proof proposes and analyzes a randomly-shifted variant of Ben-Eliezer's generic tester for local properties \cite{BE19}, and shows that this more structured tester can be simulated by a pair tester.

We show that the pair tester is online erasure-resilient (and corruption-resilient) in a strong sense: 
with good probability it never queries previously manipulated
elements. The analysis distinguishes between two types of erasures: ones that happen after the first element of a pair is queried, but before the second element; and all other erasures. Roughly speaking, the sharp distinction in the erasure thresholds between batch size $b=1$ and $b=2$ exists since the first type of erasures, which is only available to the adversary when $b=1$, is substantially more effective than the second type in hiding information.

\subsection{Related Work on  Erasures and Corruptions in Property Testing}\label{sec:related-work}
Erasure-resilient testing was first investigated by Dixit et al.~\cite{DixitRTV18} in an offline model.
In the model of Dixit et al., also studied in~\cite{RV18,RRV19,BenFLR20,PallavoorRW22,LPRV21,NV20},
 the adversary performs all erasures to the function before the execution of the algorithm. 

As discussed, the online testing model was defined in \cite{KalemajRV22}. 
In addition to the results already mentioned,
\cite{KalemajRV22} gave an online $\eps$-tester for quadraticity of functions $f:\{0,1\}^n\to\{0,1\}$ that has query complexity $O(1/\eps)$ for constant erasure rate $t.$ The dependence on $t$ in the query complexity was doubly exponential, and it was left open to obtain a quadraticity tester that can deal with corruptions. 
In a concurrent work, 
 Minzer and Zheng~\cite{MinZ} improve the quadraticity tester and vastly generalize it to deal with all low-degree properties $\cP_d$ over general fields $\F_q$. Their tester works in the presence of
 $t$-online fixed-rate erasure adversaries
 and makes $O\left(\frac 1 \eps \log^{3d+3}\frac t \eps\right)$ queries when $q$ is a prime and 
 $q^{O(1)}\cdot O\left(\frac 1 \eps \log^{3d+3q}\frac t \eps\right)$ queries when $q$ is a prime power.

\paragraph{Testing in dynamic environments.}
Another related line of work is on property testing in dynamic environments (see, e.g., \cite{goldreich2017learning, NakarR21}),
which considers settings where the tested object undergoes changes independent of the actions taken by the tester. 
The (adversarial) online testing model extends the study of dynamic settings to regimes where the data in the object, or the access to the data, may change continuously in response to the actions of the tester.

\subsection{Connection to Maker-Breaker Games}\label{sec:intro-maker-breaker-games}
Positional games are a central and widely investigated topic in the modern combinatorics literature; see, e.g., the standard textbooks on this topic \cite{Beck08,HefetzKSS14}. 
Property testing in the online erasure model is closely related to positional games, and in particular to the most prominent and well studied example of positional games: \emph{Maker-Breaker games}. Even  though Kalemaj et al.\ described their quadraticity tester as a game, they did not discuss the connection to the Maker-Breaker literature. 

A Maker-Breaker game is defined by a finite set $X$ of board elements and a family $\cW\subseteq 2^X$ of winning sets. In an $(s : t)$ Maker-Breaker game, two players, called Maker and Breaker, take turns in claiming previously unclaimed elements of $X$. On each turn, Maker claims $s$ board elements, whereas Breaker claims $t$ elements. Maker wins the game if she manages to claim all elements of some winning set; otherwise, Breaker wins.

In online testing, the algorithm plays the role of Maker and the adversary is Breaker. The set $X$ is the domain of the function, and the winning sets are witnesses, i.e., tuples of points that demonstrate that the function does not have the property. A big complication is that the tester does not know in advance which sets are in $\cW$. A prerequisite for designing an online tester is being able to identify the general structure of the sets in $\cW$ and a winning strategy for Maker. For example, Kalemaj et al.\ build their tester for quadraticity by identifying a winning strategy for a game where $\cW$ consists of ``cubes'' of the form $(x,y,z,x+y,x+z,y+z,x+y+z)$. 
Our low-degree tester (for the special case of quadraticity) uses more intricate winning sets; see the discussion of patterns in \Cref{sec:low-degree}. 

Note that the original online model corresponds to $(1:t)$ Maker-Breaker games, whereas the version with batches of size $b$ corresponds to general $(b:t)$ Maker-Breaker games. This provides additional motivation to study batches.  Going in the other direction, we hope that online testing inspires new research on Maker-Breaker games. A lot of the current literature on positional games focuses on the case where the board is a complete graph and the winning sets are graph-related (e.g., cliques). Examples from online property testing may motivate new research on Maker-Breaker games with emphasis on other types of boards, such as the hypercube.

\section{Preliminaries}
\paragraph{Notation.} 
We use $[n]$ to represent $\set{1,2,\dots,n}$ and $\log$ to denote logarithm base 2. 
Let $\cP_d$ be the class of Boolean functions $f:\set{0,1}^n\to\set{0,1}$ of degree at most $d$ over $\F_2$.

\subsection{Online Testing}\label{sec:model-definitions}
We model access to the input with a sequence $\set{\oracle{i}}_{i\in {\mathbb N}}$ of oracles, where $\oracle{i}$  is used to answer the $i$-th query (or, more generally, the $i$-th batch of queries). Oracle $\oracle{1}$ gives access to the original input (e.g., when the input is a function $f$, we have $\oracle{1} \equiv f$). Subsequent oracles are objects of the same type as the input (e.g., functions with the same domain and range). Each such oracle is obtained by the adversary by modifying the previous oracle to include a growing number of erasures/corruptions as $i$ increases.
We use $\Dist(\oracle{},\oracle{}')$ for the Hamming distance between the two oracles (i.e., the number of queries for which they give different answers).
We let $t \in \R_{\geq 0}$  denote the number of \emph{erasures} (or \emph{corruptions}) \emph{per query} (or a batch of queries). 

\begin{definition}[Fixed-rate and budget-managing adversaries]\label{def:fixed-rate_budget-managing}
     Fix a parameter $t>0$. A sequence\footnote{Our algorithms only access a finite subsequence of this sequence.} of oracles $\oracle{} =  \set{\oracle{i}}_{i\in {\mathbb N}}$
     is induced by a \emph{$t$-online fixed-rate adversary} if $\oracle{1}$ is equal to the input and, for all $i\in \N$, 
     \[
        Dist(\oracle{i}, \oracle{i+1}) \leq \floor{(i+1)\cdot t} - \floor{i\cdot t}.
     \]
    A sequence of oracles $\oracle{} =  \set{\oracle{i}}_{i\in {\mathbb N}}$
     is induced by a \emph{$t$-online budget-managing adversary} if $\oracle{1}$ is equal to the input and, for all $i \in \N$,
     \[
        \Dist(\oracle{1}, \oracle{i+1}) \leq i\cdot t.
     \]
\end{definition}
     
Note that a budget-managing adversary has more power: 
an oracle sequence that can be induced by a $t$-online fixed-rate adversary can also be induced by a $t$-online budget managing adversary.

\begin{definition}[Batch-$b$ adversary] Fix $b\in\N.$
A \emph{batch-$b$ adversary} uses oracle $\oracle{1}$ to answer the first $b$ queries and, more generally, oracle $\oracle{i}$ to answer the $i$-th batch of $b$ queries. By default (if $b$ is not specified), we assume $b=1.$
\end{definition}
 Our complexity measure is always the total number of queries, regardless of the batch size $b$. Finally, we consider two types of manipulations to the input: erasures and corruptions.
\begin{definition}[Erasure and corruption adversaries]
Let $\perp$ represent the erasure symbol. A sequence of oracles $\oracle{} =  \set{\oracle{i}}_{i\in {\mathbb N}}$ is induced by an \emph{erasure adversary} if for all $i\in\N$ and data points~$x$,
$$\oracle{i+1}(x) \in \set{\oracle{i}(x), \perp} .$$
A \emph{corruption adversary} can change answers to anything in the range, i.e., $\oracle{i}$ can be any valid input for the computational task at hand.
\end{definition}

A property $\cP$ denotes a set of objects (typically, a set of functions). Intuitively, it represents the set of positive instances for the testing problem. The (relative Hamming) distance between a function $f$ and a property $\cP$, denoted $dist(f,\cP)$, is the smallest fraction of function values of $f$ that must be changed to obtain a function in $\cP.$ Given a proximity parameter $\eps\in(0,1)$, we say that $f$ is \emph{$\eps$-far} from $\cP$ if $dist(f,\cP)\geq\eps.$
 An online tester is given a proximity parameter $\eps$ and, in addition, one or two parameters that characterize its adversary: the rate of erasures (or corruptions) $t$ and (optionally) the batch size $b.$
\begin{definition}[Online $\eps$-tester]\label{def:online_tester} Fix $\eps\in(0,1).$
    An online $\eps$-tester $\cT$ for a property $\mathcal{P}$ that works in the presence of a specified adversary (e.g., $t$-online batch-$b$ erasure budget-managing adversary) is given access to to an input function $f$ via a sequence of oracles 
    $\oracle{} =  \set{\oracle{i}}_{i\in {\mathbb N}}$
     induced by that type of adversary.
    For all adversarial strategies of the specified type,
    \begin{enumerate}
        \item if $f \in \mathcal{P}$, then $\cT$ accepts with probability at least 2/3, and
        
        \item if $f$ is $\eps$-far from $\cP,$
        then $\cT$ rejects with probability at least 2/3, 
    \end{enumerate}
    where the probability is taken over the random coins of the tester. If $\cT$ works in the presence of an erasure (resp., corruption) adversary, we refer to it as an online-erasure-resilient (resp., online-corruption-resilient) tester.
    
    If $\cT$ always accepts all functions $f\in\cP$, then it has \emph{1-sided error.} If $\mathcal{T}$ chooses its queries in advance, before observing any outputs from the oracle, then it is \emph{nonadaptive}.
\end{definition}

To ease notation, we use $\oracle{}(x)$ for the oracle's answer to query  $x$  (omitting the timestamp $i$). 
If $x$ was queried multiple times, $\oracle{}(x)$ denotes the first answer given by the oracle.
 
\section{Linearity Testing}\label{sec:linearity}

This section is dedicated to proving \Cref{thm:linearity_tester}. We first analyze an \emph{offline} tester which we later simulate as  part of our main algorithm.
Since it uses Fourier analysis, here we use the standard notation switch for the value returned by a Boolean function (see, e.g., \cite{OD1}), where true is denoted by -1 and false is denoted by 1. As a result, the input function is of the form $f\colon\{0,1\}^n\to\{-1,1\}$, and it is linear if  $f(x_1) \cdot f(x_2) = f(x_1\xor x_2)$ 
 for all $x_1,x_2\in\cube{n}.$
\subsection{An Offline
Linearity Test}
\label{sec:linearity_generalized_patterns}
We first define $\blr{k}$, introduced by~\cite{KalemajRV22}, which naturally extends the BLR test.
\begin{algorithm}
 \caption{$\blr{k}$
 }
 \label{alg:krv-linearity-test}
{
    \Given {Even integer parameter $k\geq 2$ and query access to a function $f:\{0,1\}^n\to\{-1, 1\}$ 
    }
     Query $k$ points $x_1,\dots, x_k \in \cube{n}$ chosen uniformly at random (with replacement).\\
     Query point $y = \bxor_{i\in[k]} x_i$.\\
     {\bf Accept} if $f(y) = \prod_{i\in[k]} f(x_i)$ (equivalently, if $f(y) \cdot \prod_{i\in[k]} f(x_i) = 1$); otherwise, {\bf reject}.
    }
\end{algorithm}    

The test always accepts all linear functions, as can be shown by induction on $k$. The next lemma demonstrates that $\blr{k}$ rejects functions that are $\eps$-far from linear with sufficient probability.  This is a strengthening
of \cite[Theorem 1.2]{KalemajRV22}, which bounded 
the rejection probability by~$\eps.$
\begin{lemma}     \label{lem:krv_k_test}
    If $f$ is $\eps$-far from  linear, and $k\geq 2$ is even, then
    \begin{align*}
        \Prob{}{\blr{k}(f)  \text{\em\ rejects}}
        \geq \frac{1 - (1-2\eps)^{k-1}}{2} 
        \geq \min\Big\{\frac 1 4 \, ,\, \frac{k \eps}{2}\Big\}.
    \end{align*}   
\end{lemma}
 Since $k \geq 2$ and $\eps\in(0,1/2]$, we get $(1-2\eps)^{k-1} \leq 1-2\eps$, with the first inequality in the lemma yielding the bound proved in~\cite{KalemajRV22}.
However, as it is imperative to use $k > 2$,  the strengthened bound is crucial in obtaining an optimal online-erasure-resilient tester.
In addition, the stronger guarantee we prove suffices to obtain new optimal offline linearity testers by repeating Algorithm~\ref{alg:krv-linearity-test}; this can be done with any value of $k$ from $3$ to $O(1/\eps)$.

\begin{proof}[Proof of \Cref{lem:krv_k_test}.]
The key tool used in the proof is Fourier analysis (see, e.g., \cite{OD1} for an overview of the technique and standard facts). We start by giving a couple of standard definitions.

The \emph{character} functions $\chi_S \colon \{0, 1\}^n \to \{-1,1\}$, defined as $\chi_S = (-1)^{\sum_{i \in S} x_i}$ for $S \subseteq [n]$, form an orthonormal basis for the space of  all real-valued functions on $\{0,1\}^n$ equipped with the inner-product 
$\langle g, h \rangle = \underset{x\sim \{0,1\}^n}{\E}[g(x)h(x)],$
where $g, h \colon \{0,1\}^n \to \R$. 
For $g \colon \{0,1\}^d \to \R$ and $S \subseteq [n]$, the Fourier coefficient of $g$ on $S$ is 
$
    \widehat{g}(S) = \langle g, \chi_S \rangle = \underset{x\sim \{0,1\}^n}{\E}[g(x)\chi_S(x)].
$

Now consider a function $f:\{0,1\}^n \to \{-1,1\}$ that is $\eps$-far from linear. It is well known that the distance from $f$ to linearity is $\frac{1}{2} - \frac{1}{2}\max_{S\subseteq [n]}\widehat{f}(S).$ Since the distance is at least $\eps$, we get
\begin{align}\label{eq:distance-to-linearity}
\max_{S\subseteq [n]}\widehat{f}(S)\leq 1-2 \eps.    
\end{align}
 As shown in \cite[Equation (4)]{KalemajRV22},
\begin{align}
       \Prob{}{\blr{k}(f)  \text{\em\ rejects}}
      &= \underset{x_1, \dots, x_k \sim \{0,1\}^n}{\E} \Big[\frac{1}{2} - \frac{1}{2}
     \prod_{i\in[k]} f(x_i) \cdot f\left(\xor_{i=1}^{k} x_i\right)\Big] \nonumber \\
      &=  \frac{1}{2} - \frac{1}{2} \sum_{S \subseteq [n]} \widehat{f}(S)^{k+1}. \label{eq:KVR}
  \end{align}
 Next, we bound the sum in \eqref{eq:KVR} in terms of $\max_{S\subseteq [n]}\widehat{f}(S)$ and then apply \eqref{eq:distance-to-linearity}:
    \begin{align*}
  \sum_{S\subseteq [n]} \fc{f}{S}^{k+1}
        \leq \max_{S\subseteq [n]} \fc{f}{S}^{k-1} \cdot \sum_{S\subseteq [n]} \fc{f}{S}^{2}
        &= \max_{S\subseteq [n]} \fc{f}{S}^{k-1}
        \leq (1-2\eps)^{k-1}.
    \end{align*}
 Substituting this expression into \eqref{eq:KVR}, we obtain the first inequality in \Cref{lem:krv_k_test}.
To obtain the second inequality in \Cref{lem:krv_k_test}, we show that $1 - (1-2\eps)^{k-1}$ is at least $1/2$ for large values of $k$ and at least $k \eps$ for small values of $k \geq 2$. 
     The bound holds trivially for $\eps=1/2$, so from now on assume  $\eps<1/2$.
    Let $k_0 \in \R$ be such that $(1-2\eps)^{k_0 - 1}=1/2$. If $k \geq k_0$, we have     
    \[
        1 - (1-2\eps)^{k-1}\geq 1 - (1-2\eps)^{k_0-1} = 1/ 2 .
    \]   
    For the case $k\in [2,k_0)$, we prove by induction on $k$ that $1 - (1-2\eps)^{k-1}\geq k\eps$.
    For $k=2$, equality holds. For each $k\in [3, k_0)$, 
   we have
    \[
        1 - (1-2\eps)^{k-1}
        = \big(1 - (1-2\eps)^{k-2}\big) + 2\eps (1-2\eps)^{k-2}
        \geq (k-1)\eps+2\eps \cdot (1-2\eps)^{k_0 - 1}
        = k \eps,
    \]
    where for the inequality we bound the first term by the inductive hypothesis on $k-1$ and the second term using $k-2\leq k_0-1$.
 Thus, in both cases, at least one of the expressions in the minimum is a lower bound.
\end{proof}

\subsection{Online-Erasure-Resilient Linearity Tester}
\label{sec:linearity-tester}
Our algorithm (Algorithm~\ref{alg:linearity-main}) is based on multiple simulations of $\blr{k}$. Each simulation starts by querying a reserve of $m=2k$ initial points to keep many possibilities open for the final $\blr k$ query.
Specifically,
a reserve of $2k$ random points creates roughly $2^{2k}$ different options for the final $\blr k$ query\footnote{In \cite{KalemajRV22}, the reserve is used to simulate
$\blr k$ with different values of $k$. 
Fixing the value of $k$ to half of the reserve size eases our analysis (allowing us to use \Cref{lem:krv_k_test} with the same $k$)  while providing many options.}.
We set $m$ to about $\log t$ to ensure that a significant fraction of options remains open before the final query of the \blr{k} simulation is made. There is also a small additional dependence on $\eps$ to overcome the fact that the overall number of iterations depends on $\eps$ and, as a result, for some settings of parameters, the probability of error in each iteration has to be proportional to $\eps.$
In principle, setting $m = \Theta(\log t +1/\eps)$ and the number of iterations, $r$, to $\Theta(1)$ is enough for the desired query complexity, but it significantly restricts the range of parameters (e.g., it only works when
$\eps \geq 2/n$). 
 Our choice of parameters is more subtle: besides ensuring optimal query complexity, it enlarges the range of $\eps$ and $t$ (in \Cref{thm:linearity_tester}) for which Algorithm~\ref{alg:linearity-main} is applicable. To do this, we set $m$ to a smaller value: $\log t +o(\log t) + \Theta(\log(1/\eps))$, and the number of iterations, $r$, depends on $\eps$ and $m$.
\begin{algorithm}
 \caption{Online-Erasure-Resilient Linearity Tester}\label{alg:linearity-main}
{
    \Given {Parameters $\eps\in(0,1/2],t\in\N$; query access to $f$ via $t$-erasure oracle sequence $\oracle{}$ 
    }
    $t\gets\max\{t,2\}$
    \Comment{If $t< 2$, replace it with $t=2.$
    }\\
    $m \gets 4\Big\lceil{\frac{1}{4}\big(14+\log t+\log\log^2 t+\log \frac{1}{\eps^2}}\big)\Big\rceil$, $\alpha\gets\min\big\{\frac{1}{4},\frac{m\eps}{4}\big\}$ and,
    $r \gets \ceil{\frac{5}{4\alpha}}$.\\
     \RepTimes{$r$}  {\label{step:linearity-repeat}
         Sample $X =(x_1,\dots,x_{m})\in(\cube{n})^m$ uniformly 
         at random.\\
         Query $f$ at points $x_1,\dots, x_m$.\\
         Query $f$ at point $y = \xor_{j\in S} x_j$, where $S$ is a uniformly random subset of $[m]$ of size $\frac m2$.\\ 
         \If{$\oracle{}(y) \cdot \Pi_{j\in S} \oracle{}\left(x_j\right) = -1$}{\textbf{Reject}
         \Comment{This implies no erasures in this iteration.}
         }
        
    }
        \textbf{Accept}
  }  
 \end{algorithm}
Crucially, in each iteration, the distribution over the queries made by Algorithm~\ref{alg:linearity-main} is identical to that in the $\blr{k}$ test. Indeed, as $S$ and $X$ are independent, we could draw $S$ first. Then, for each choice of $S$, the marginal distribution of $\set{x_j}_{j\in S}$ is uniform over $\left(\cube{n}\right)^{k}$. Finally, the last query is  $y = \bxor_{j\in S} x_j$, resulting in 
the same distribution on the $k+1$ queries as in $\blr{k}$. 

The second ingredient we need is the following lemma.
\begin{restatable}{lemma}
    {LinearityNoErasuresLemma}\label{claim:linearity_no_erasures}
    For all $m\geq 10$, the probability that one specific iteration of the loop in Line~\ref{step:linearity-repeat} of Algorithm~\ref{alg:linearity-main} queries an erased point is at most $3\alpha/25$. 
\end{restatable}

We next \short{}{restate and }prove the main theorem of this section,
deferring the proof of \Cref{claim:linearity_no_erasures} to the next section.
\short{}{
\OptimalOnlineLinearityTesterThm*}
\begin{proof}[\short{Proof of \Cref{thm:linearity_tester}}{Proof}]
We first focus on erasures, and show that Algorithm~\ref{alg:linearity-main} satisfies the conditions of the theorem. Clearly, it always accepts all linear functions.
Now, fix an adversarial (budget-managing) strategy and suppose that the input function is $\eps$-far from linear. By \Cref{lem:krv_k_test} and since $k=m/2$ in Algorithm~\ref{alg:linearity-main} is even, the probability that one iteration of the loop in Step~\ref{step:linearity-repeat}  samples a witness of nonlinearity is at least $\alpha=\min\set{1/4,m\eps/4}.$ 

  By \Cref{claim:linearity_no_erasures}, the probability that an erasure is seen in a specific iteration is at most $\frac{3\alpha}{25}$.
     By a union bound, the probability of a single iteration seeing an erasure or not selecting a witness of nonlinearity is at most $1-\alpha + \frac{3\alpha}{25} = 1 - \frac{22\alpha}{25}$.
    Algorithm~\ref{alg:linearity-main} errs only if this occurs in all iterations. By independence of random choices in different iterations, the failure probability is at most
    \[
    \Big(1-\frac{22\alpha}{25}\Big)^r\leq 
        \Big(1-\frac{22\alpha}{25}\Big)^{\frac{5}{4\alpha}}
        \leq e^{-1.1}
        \leq \frac{1}{3},
    \]
    where we used that $r=\ceil{\frac{5}{4\alpha}}\geq \frac{5}{4\alpha}$ and 
    $1-x\leq e^{-x}$ for all $x$.
  The query complexity 
  is 
  at most
  \[
  r(m+1)
  \leq \left(\frac{5}{4\alpha}+1\right)(m+1)
  \leq \frac{4m}{\alpha}
  = \max\set{16m, \frac{16}{\eps}}
  \leq \max\set{640 \cdot \log t, \frac{640}{\eps}},
  \] 
  where 
  the last inequality is due to the fact that $m\leq 20(\log t+1/\eps)$.

To show that 
Algorithm~\ref{alg:linearity-main} is also corruption-resilient, one may apply  \cite[Lemma 1.8]{KalemajRV22},
which essentially says an error-resilient algorithm that has probability at most $1/3$ to either err or see a manipulation, is also corruption-resilient.
 For the soundness, our analysis holds since it suffices to have one iteration that finds a witness without seeing manipulations. For completeness, note that the algorithm can only err if it has seen a manipulation, and the probability of seeing a manipulation at any iteration is at most $3\alpha/25$ by \Cref{claim:linearity_no_erasures}. Using a union bound, the overall probability of seeing any manipulated entry during the entire execution is at most  
\[\frac{3\alpha}{25} \cdot r \leq \frac{3\alpha}{25}\left(\frac{5}{4\alpha}+1\right)=\frac 3{20}+\frac{3\alpha}{25}\leq \frac 1 3 .\qedhere\]
\end{proof}

\begin{remark}
Our tester is applicable for $t \leq \poly(\eps) \cdot 2^{n/2}$.
On the other hand, it can be easily shown that for $t = \Omega\left(\eps 2^{n}\right)$ testing is impossible.  This follows from the fact that there exists some constant $c$ such that $c/\eps$ queries are not enough, and by then the adversary can already manipulate all other entries.
We leave it as an open question to understand for which values of $t$ linearity can be tested in the online setting. Such questions are investigated for other properties in \Cref{sec:local}.
\end{remark}
\subsection{Probability of Seeing an Erasure}
\label{sec:earasure_probability}
Recall that \Cref{lem:krv_k_test} shows that, assuming $f$ is $\eps$-far, the probability of spotting a witness in a single iteration is at least $\alpha=\min\{\frac{1}{4},\frac{m\eps}{4}\}$.
In this section we prove the probability of querying an erasure is at most $3\alpha/25$ in every iteration, as stated in \Cref{claim:linearity_no_erasures}.
We start with an auxiliary claim, similar to an argument that appears in the proof of~\cite[Lemma 2.8]{KalemajRV22}.
\begin{claim}
    \label{claim:no_collisions_for_y}
    Fix an iteration of the loop in Step~\ref{step:linearity-repeat} of Algorithm~\ref{alg:linearity-main}. 
    For all $T\subseteq[m]$, 
    let $y_{_T} = \bxor_{j\in T} x_j$.
    Then, the probability there exist two subsets $T_1\neq T_2$ of the set $[m]$ with $y_{_{T_1}} = y_{_{T_2}}$ is small:
    \[
        \Prob{X}{\exists T_1\neq T_2 \text{\rm\ such that\ } y_{_{T_1}} = y_{_{T_2}}} \leq \frac{\alpha}{25} .
    \]
\end{claim}

\begin{proof}

We first bound
$m$ using its setting in Algorithm~\ref{alg:linearity-main} and the premise $t \cdot \log^2 t \leq 2^{-21}\cdot\eps^{2.5} 2^{n/2}$ from \Cref{thm:linearity_tester}:
\begin{align}\label{eq:linearity-bound-on-m}
        m
        \leq \log\left(\frac{2^{18} t  \log^2 t}{\eps^2}\right) 
        \leq \log\left(\frac{2^{-3}\eps^{2.5}2^{n/2}}{\eps^2}\right)
        \leq
\log\left(\sqrt{\eps \cdot 2^n/50}\right)
        = \frac{\log\left(\eps \cdot 2^n/50\right)}{2}.
    \end{align}
   Consider two distinct sets $T_1,T_2\subset [m]$. W.l.o.g.\ there exists an element $\ell\in T_1 \setminus T_2$. Fix all entries in $X$ besides $x_\ell$. The value of $y_{_{T_2}}$ is now fixed, but over the random choice of $x_\ell\in\cube{n}$, the vector $y_{_{T_1}}$ is uniform over $\cube{n}$. Thus, $\Pr_{x_{\ell}}[y_{_{T_1}} = y_{_{T_2}}]=2^{-n}$ and, consequently,
    \[
        \Pr_X[y_{_{T_1}} = y_{_{T_2}}]
        = \E\big[\Pr_{x_{\ell}}[y_{_{T_1}} = y_{_{T_2}}]\big]
        =\E[2^{-n}] = 2^{-n},
    \]
    where both expectations are over all entries in $X$ besides $x_\ell$, which are drawn independently from $x_{\ell}$.
    We use a union bound over all pairs of subsets $T_1$ and $T_2$,  and then apply \eqref{eq:linearity-bound-on-m}
    to get 
    \[
        \Prob{X}{\exists T_1\neq T_2 \text{\rm\ such that\ } y_{_{T_1}} = y_{_{T_2}}}
        \leq \frac{2^{2m}}{2^n}\leq \frac{\eps}{50}
        \leq \frac{\alpha}{25} 
        .
    \]
    The last inequality holds since $\eps\leq \min\{1/2,m\eps/2\}=2\alpha$ for all $m\geq 2$. 
\end{proof} 

We now upper bound the probability of seeing an erasure 
in one iteration.

\begin{proof}[Proof of \Cref{claim:linearity_no_erasures}]
 
The total number of erasures performed during the execution of the algorithm is at most $tr(m+1)$.
We define three bad events and give upper bounds on their probabilities.
    \paragraph{An erasure while querying $X$.} Let $B_1$ be the event that $\oracle{}(x_j) = \perp$ for some point $x_j$ sampled in this iteration, where $j\in[m]$. Each point $x_j$
   is sampled uniformly from $\cube{n}$, so the probability it is erased is at most $tr(m+1)/2^n$. By a union bound over all $m$ points, recalling $n\geq 2m$, we get 
    \begin{equation}
        \label{eq:bound_b1}
        \Pr_X[B_1]\leq \frac{tr(m+1)m}{2^{n}} 
        \leq \frac{tr(m+1)m}{2^{2m}}.
    \end{equation}
   
    \paragraph 
    {$X$ induces a bad distribution of $y$ points.}
    Let $B_2$ be the event that, in this iteration, there exist two different choices of $S$ leading to the same choice $y \in \cube{n}$. 
    By Claim~\ref{claim:no_collisions_for_y}, $\Pr[B_2]\leq\frac{\alpha}{25} $. 
    
    \paragraph{An erasure on query $y$.} Let $B_3$ be the event that $\oracle{}(y) = \perp$ for $y$ queried in this iteration.
    The adversary knows $X$ before $y$ is queried, but there are plenty of choices for $y$.
    Conditioned on $\overline{B_2}$, the distribution of $y$ 
    is uniform over  
    $\binom{m}{m/2}$ different choices. 
     We use 
    $\binom{m}{m/2} \geq 2^m/\sqrt{2m}$, to obtain
\begin{equation}
    \label{eq:bound_b3}
        \Pr[B_3 | \overline{B_2}]
        \leq \frac{tr(m+1)}{\binom{m}{m/2}}
        \leq \frac{tr(m+1)\sqrt{2m}}{2^m}.
\end{equation}
 
In terms of the bad events, the lemma states that $\Pr[B_1\cup B_3]\leq 3\alpha/25.$
By using a union bound over $B_1$ and $B_3$ and then the law of total probability to compute $\Pr[B_3]$, we get
   \begin{align*}
       \Pr[B_1\cup B_3]
       &\leq \Pr[B_1]+\Pr[B_3 | \overline{B_2}]\cdot\Pr[\overline{B_2}]+\Pr[B_3 | B_2]\Pr[B_2]      \leq \underbrace{\Pr[B_1]+\Pr[B_3 | \overline{B_2}]}_{(\star)}+ \Pr[B_2]. 
   \end{align*}
  
Since $\Pr[B_2] \leq \frac{\alpha}{25}$, it remains to show  
that $(\star)\leq \frac{2\alpha}{25}$, 
in order to complete the proof. To do this, we combine the bounds from \eqref{eq:bound_b1} and \eqref{eq:bound_b3}: \begin{equation}\label{eq:seeing_erasure_bound_1}
       (\star)= \Pr[B_1]+\Pr[B_3 | \overline{B_2}]\leq \frac{tr(m+1)m}{2^{2m}}+\frac{tr(m+1)\sqrt{2m}}{2^m}\leq \frac{tr m^2}{2^{m+1}}\leq  \frac {r m^2 \eps^2}{2^{15}\log^2 t}  .
    \end{equation} 
    For the second inequality, we use the fact that 
   $\frac{(m+1)m}{2^m}+(m+1)\sqrt{2m}\leq \frac{m^2}{2}$ for all $m\geq 10$. For the last inequality, we use the value of $m$ from Algorithm~\ref{alg:linearity-main} which satisfies $2^m\leq \frac{2^{14}t\log^2t}{\eps^2}$.
To further bound \eqref{eq:seeing_erasure_bound_1}, we split the analysis into two cases, depending on the value of $\alpha=\min\{1/4,m\eps/4\}$.

  \textbf{Case I: $\alpha=\frac{m\eps}{4}$. } 
    In this case, 
   there are $r =\ceil{\frac{5}{m\eps}}\leq \frac{5}{m\eps}+1$ iterations, 
   which means that $rm\eps \leq 5+m\eps\leq 6$. 
    Therefore, 
    \[(\star)\leq 
    \frac{(r m \eps)m \eps}{2^{15}\log^2 t} 
    \leq
    \frac{6 m \eps}{2^{15}\log^2 t} 
    \leq \frac { m \eps}{50}=\frac{2\alpha}{25}
    .\]
    
  \textbf{Case II: $\alpha=\frac 1 4$. }
        For this case, note that
   \begin{align*}
           m 
          &\leq 18 + \log t + 2\log\log t + 2\log(1/\eps)\\
          &\leq 18+2.1\log t+1.1/\eps\leq 11.2\cdot\max\{2\log t,1/ \eps\},
        \end{align*}
        where the second inequality uses $2\log y\leq 1.1y$ for all $y>0$ (with $y=1/\eps$ and $y=\log t$), and the last inequality uses $2\log t \geq 2$ and $1/\eps\geq 2$.
        It follows that $$\frac{m\eps}{2\log t}\leq 11.2 \frac{\max\{2\log t,1/\eps\}}{2\log t\cdot 1/\eps}=\frac{11.2}{\min\{2\log t,1/\eps\}} \leq 5.7.$$
       Finally, we use $r=\ceil{\frac{5}{4\alpha}}=5$ to obtain from \Cref{eq:seeing_erasure_bound_1}
        \[
        (\star)\leq  
        \frac{r}{2^{13}}\left(\frac { m \eps}{2\log t}\right)^2 
             \leq  \frac {5\cdot (5.7)^2}{2^{13}} 
             \leq\frac{1}{50}=\frac{2\alpha}{25} .
         \qedhere\] 
\end{proof}

\section{The Lower Bound for Low-Degree Testing}\label{sec:low-degree-lb}

In this section, we prove our lower bound for online low-degree testing\short{ stated in \Cref{thm:degree-d-lb}.}{.
\LowerBoundForLowDegreeTestingThm*}
\begin{proof}[\short{Proof of \Cref{thm:degree-d-lb}}{Proof}]
    Fix a degree $d > 1$. Let $\binom{n}{\leq d}$ denote $\sum_{i=0}^d \binom n i$.
    For a vector $x \in \F_2^n$, define its extension $\ext(x)\in \F_2^{\binom{n}{\leq d}}$ to be the vector where each entry is indexed by a set $S\subseteq [n]$ of at most $d$ coordinates and $\ext(x)_{S} = \prod _{i\in S} x_i$. In other words, each entry of $\ext(x)$ is an evaluation of one monomial of degree at most $d$ on input $x$.
    The extended vector $\ext(x)$ includes all entries $x_i$ of $x$ (as entries indexed by singletons $S = \set{i}$). We define projection $\pi$ as the inverse of $\ext$, i.e., 
     $\pi(\ext(x)) = x$.
    
    Consider the following strategy $\cS$ for a fixed-rate erasure adversary. Suppose the tester has successfully obtained answers to $k$ distinct
    queries $y_1, \dots, y_k$.
   Consider the
    linear space in $\F_2^{\binom{n}{\leq d}}$ spanned by $\ext(y_1), \dots, \ext(y_k)$, and 
     let $Z\subseteq \F_2^n$ denote set of projections of its vectors back to $\F_2^n$  using $\pi$.
    Note that $y_i \in Z$ for $i\in[k]$.
    The adversary simply tries to erase all (non-erased, non-queried) points of $Z$. Observe that $Z$ is determined by the queries made by the tester and does not depend on the input function $f$.
We next show that $\card{Z}$ is small relative to $k$, and so unless $k$ is large enough, the adversary can entirely erase $Z$. 
    \begin{claim}\label{claim:low-degree-lb}
        Let $r =|Z|.$ Then $\binom{\floor{\log r}}{d}\leq k$.
    \end{claim}
 \begin{proof}  
    Consider the matrix $A$ with rows indexed by polynomials of degree at most $d$ and columns indexed by vectors $z\in Z$, where the entry indexed by a polynomial $p$ and a point $z$ contains the value $p(z) \in \F_2$. To prove the claim, we give two bounds on the rank of $A.$

    \paragraph{Upper bound: $rk(A) \leq k$.} Let $A'$ be the submatrix of $A$ with rows indexed by monomials. Since every degree-$d$ polynomial is a linear combination of monomials, we have $rk(A) = rk(A')$.
    In $A'$, the column indexed by $z$ is exactly the extended vector $\ext(z)$. By definition, every column in $A'$ is spanned by the $k$ columns indexed by $y_1, \dots, y_k$, showing that $rk(A') \leq k$.

    \paragraph{Lower bound: $rk(A) \geq \binom{\floor{\log r}}{d}$.}
    Let $B$ be an arbitrary submatrix of $A$ with only $r' \leq r$ columns, where $r' = 2^a$ is the largest power of $2$ not exceeding $r$.
    Trivially, $rk(A) \geq rk(B)$.
    We now apply \cite[Lemma 1.4]{Ben-EliezerHL12} (or  \cite[Theorem 1.5]{KeevashS05}) to see that the dimension of the rows of $B$ is at least
    \[
        \binom{a}{\leq d}
        \geq \binom {a}{d}
        = \binom{\floor{\log r}}{d} .
    \] 
     The two inequalities together yield the claim. \end{proof}

    Now we apply Yao's minimax principle for the online model \cite[Corollary 9.4]{KalemajRV22}. It states that
    to prove a lower bound $q$ on the worst-case query complexity of online-erasure-resilient property testing,
		it suffices to give an adversarial strategy $\cS$,
	a distribution $\mathcal{D}^+$ on positive instances, 
 and a distribution $\mathcal{D}^-$ on instances 
 that are negative with probability at least $\frac{6}{7}$,	such that  
  every deterministic $q$-query algorithm that accesses its input via an oracle using strategy $\cS$ sees the same distribution on the the query-answer histories under $\cD^+$ and $\cD^-$. 

Let $\cD^+$ be the uniform distribution over all degree $d$ Boolean functions on $\{0,1\}^n$ and  $\cD^{-}$ be the uniform distribution over all Boolean functions on $\{0,1\}^n$. 
 
 We show that a function $f \sim \cD^-$ is $\frac 1 3$-far from $\cP_d$ (set of functions of degree at most $d$) with probability at least $6/7$.
 Let $g\in\cP_d, f \sim \cD^-,$ and $\dist(f,g)$ be the fraction of domain points on which $f$ and $g$ differ. Then, $\E[dist(f,g)] = 1/2$. By the Hoeffding bound, $\Pr_{f\sim \cD^-}[\dist(f,g) \leq \frac{1}{3}] \leq e^{-\frac{2^n}{18}}$. 
By a union bound over the $2^{\binom{n}{\leq d}}$ functions of degree at most $d$, we get $\Pr_{f\sim \cD^-}[dist(f,\cP_d)\geq \frac 13] \geq 1 - 2^{\binom{n}{\leq d}} \cdot e^{-\frac{2^n}{18}}$. For large enough $n$, this probability is at least $6/7$.

 Let $A$ be a deterministic algorithm that makes $q \leq \binom{\floor{\log t} - 1}{d}$ queries to the oracle $\cO$ with adversarial strategy $\cS$. By \Cref{claim:low-degree-lb}, we get $r \leq t$, i.e., after each query $y_i$ for $i\in[q]$, the adversary has sufficient erasure budget to erase $Z$. 
 We argue that the distributions on the histories of query answers are the same under the distributions $\cD^+$ and $\cD^-$.
 Under both distributions, the adversary erases the same query points. Consider a query $y_k$ of $A$ that is not erased and suppose it is made after $A$ successfully obtained answers $a_1,\dots,a_{k-1}$ to distinct queries $y_1,\dots,y_{k-1}$. When $f\sim \cD^-$, we have $\Pr_{f\sim \cD^-}[f(y_k)=a_k \mid f(y_1)=a_1 \wedge\dots \wedge f(y_{k-1})=a_{k-1}]=1/2$ for all $a_1,\dots,a_k\in\{0,1\}$, since the value of $f(y_k)$ is set uniformly and independently of values at other points. 
Now consider $f\sim \cD^+$, viewed as a uniformly random coefficient vector $\ind{f}\in\set{0,1}^{\binom{n}{\leq d}}$, where each entry corresponds to a monomial of degree at most $d$. For any query it holds that $f(y_i) = \ip{\ind{f}}{\ext(y_i)}$.
By definition of $\cS$, the extension $\ext(y_k)$ is linearly independent of $\ext(y_1), \dots, \ext(y_{k-1})$, which implies that for any fixed values $b_1, \dots b_{k-1}$ of $\ip{\ind{f}}{\ext(y_i)}, \dots, \ip{\ind{f}}{\ext(y_{k-1})}$, the value of $\ip{\ind{f}}{\ext(y_k)}$ is still uniformly distributed. 

Thus, $\cD^+$ generates degree-$d$ polynomials, $\cD^-$ generates functions that are $\frac 13$-far from $\cP_d$ with probability at least $\frac 67$, and the query-answer histories for any deterministic algorithm $A$ that makes $q\leq\binom{\floor{\log t} - 1}{d}$  queries and runs against the $t$-online fixed-rate erasure oracle employing strategy $\cS$ are identical under $\cD^+$ and $\cD^-$. Consequently,
 Yao's principle implies the desired lower bound. 
 \end{proof}

\section{Low-Degree Testing}\label{sec:low-degree}
This section is dedicated to the proof of \Cref{thm:d_deg:chain_of_cubes} that gives an online tester for the property $\cP_d$ (being a polynomial of degree at most $d$) with batches of size $b=2^{d-1}.$

Our strategy is a natural extension of the one employed for linearity in \Cref{sec:linearity-tester}: we view the 
\blr{2}
as a ``linear square" ($x, y, x+y$, omitting the origin), and perceive the extended test as a ``chain of squares" (see the top part of \Cref{fig:chain_of_cubes}). For quadraticity and any degree $d > 2$, we resort to ``chains of cubes" in a similar way (see the bottom part of \Cref{fig:chain_of_cubes}).

\begin{figure}[!t]
\centering
 \includegraphics[scale=1.3]{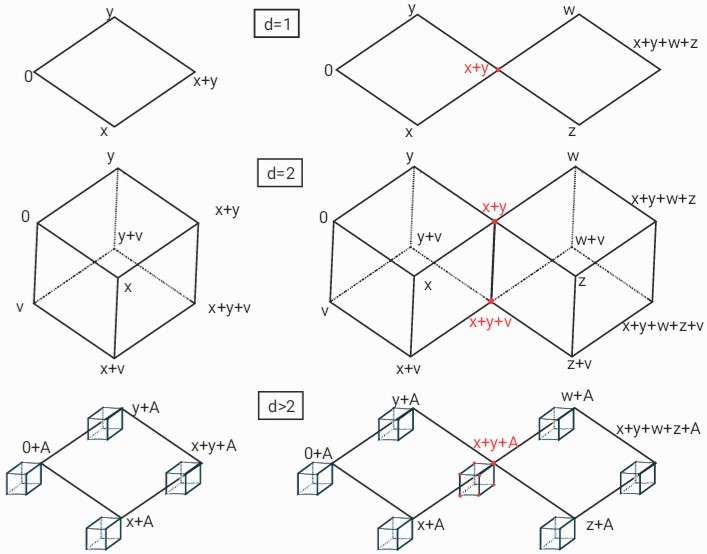}
  \caption{Chains of cubes (with 0 as the first parameter). The left side shows linear subspaces (drawn as cubes); the right side shows chains formed by two cubes. Red points overlap and cancel each other. From top to bottom, we have witnesses for linearity (chains of squares), witnesses for quadraticity (each vertical edge represents a batch of $b=2$ points), and witnesses for $\cP_d$,
where each batch queries an affine subspace with a fixed linear component $A$ and different translations.}
  \label{fig:chain_of_cubes}
\end{figure}

\subsection{Algebraic Testing Patterns}
In the following, we define a \emph{testing pattern} to be a mapping from a set of parameters to a structured set of points (e.g., mapping parameters $x,y$ to the three points $x, y, x+y$). Setting a value to all parameters induces an \emph{instance} of the testing pattern. To test for a low-degree property $\cP_d$, we simply draw a random instance of some ``good'' pattern and check the parity of 
the sum of function values on the points in the pattern.
This is a generalization of the tester of \cite{AlonKKLR05},
which  
checks the parity on a random linear subspace (which is an example of a good testing pattern). 

Our notion of a good testing pattern is a special case of the more abstract notion of \emph{$k$-local formal characterizations} (see  \cite[Definition 2.3]{KS08}), where each linear map defines a single testing point, and the subspace $V$ of ``good vectors'' is fixed to all answer vectors with parity $0$.
In particular, a good testing pattern is a $2$-ary independent characterization (see  \cite[Definition 2.5]{KS08}).

\begin{definition}[Testing pattern, parameter, instance]

    A \emph{testing pattern} is a matrix\footnote{We use $\mathcal{M}_{\ell \times m}(\F_2)$ for the set of $\ell\times m$ matrices with entries in $\F_2.$} 
    \[X \in \mathcal{M}_{\ell \times m}(\F_2)\] 
    with rows $c_1,\dots, c_{\ell}$, such that $\ell \geq m$, each row is unique, and $rank(X) = m$.
    An \emph{instance} of a pattern $X$ is described by the product $XM\in \mathcal{M}_{\ell \times n}(\F_2)$, for some parameter matrix $M\in\mathcal{M}_{m \times n}(\F_2)$.
    The rows of $XM$ specify $\ell$ testing points $y_1, \dots, y_{\ell}$ in $\F_2^n$. 
\end{definition}

Generalizing the notation from~\cite{AlonKKLR05}, for a function $f$ and a pattern instance $XM$, we denote the sum (over $\F_2$) of the values of $f$ on these points by
\[
    T_{X,f}(a_1,\dots, a_m) := \sum_{i\in[\ell]} f(y_i) .
\]
Each testing pattern induces a 
parity-checking 
tester, defined below.
\begin{algorithm}
 \caption{An $X$-tester}
 \label{alg:pattern-tester}
{
    \Given {a pattern $X\in \mathcal{M}_{\ell \times m}(\F_2)$ and query access to a function $f:\F_2^n\to\F_2$ 
    }
     Choose a parameter matrix $M \in \mathcal{M}_{m \times n}(\F_2)$ uniformly at random (alternatively, choose $m$ row vectors $a_1,\dots, a_m \in \F_2^n$ independently and uniformly at random).\\
     Query $f$ at testing points $y_1,\dots, y_{\ell}$, the $\ell$ row vectors of the instance matrix $XM$.\\
     {\bf Accept} if $T_{X,f}(a_1,\dots, a_m) = 0$; otherwise, {\bf reject}.
    }
\end{algorithm}

Next, we define a good testing pattern to have two desired properties.
\begin{definition}[Good patterns]\label{def:pattern_completeness_minimaly_sound}
    A pattern $X_{\ell \times m}$ is \emph{complete} for a property $\mathcal{P}$ if for all functions $f \in \mathcal{P}$,
    \[
        \Prob{a_1,\dots,a_m}{T_{X,f}(a_1,\dots,a_m) = 1} = 0 .
    \]
    A pattern $X_{\ell \times m}$ is \emph{minimally sound} for a property $\mathcal{P}$ if for all functions $f \notin \mathcal{P}$,
    \[
        \Prob{a_1,\dots,a_m}{T_{X,f}(a_1,\dots,a_m) = 1} > 0 .
    \]
     If a pattern is both complete and minimally sound for $\cP$, we call it  \emph{good} for 
    $\cP$. 
\end{definition}
In other words, a good pattern $X_{\ell \times m}$ for a property $\cP$ characterizes~$\cP$:
\[
    f \in \mathcal{P} \ \iff \ \forall a_1,\dots,a_m:\  T_{X,f}(a_1,\dots,a_m) = 0 .
\]
In the context of testing $f\in\cP_d$, we have the following useful characterization of good patterns.

\begin{restatable}{claim}{lowDegreePatternCharacterization}\label{claim:low_deg_characterization}
    A pattern $X\in\mathcal{M}_{\ell \times m}(\F_2)$ is good for $\cP_d$ if and only if the two conditions hold:
    \begin{enumerate}[(i)]
        \item \label{itm:completeness_property}For all 
        coordinate sets $S \subseteq [m]$ of size $\card{S} \leq d$, we have
                   $\ \displaystyle\sum_{i\in[\ell]} \prod_{j\in S} X_{ij} = 0 .$
        \item\label{itm:soundness_property} There exists a set $S \subseteq [m]$ of size $\card{S} = d+1$ such that
            $\ \displaystyle\sum_{i\in[\ell]} \prod_{j\in S} X_{ij} = 1 .$
    \end{enumerate}
\end{restatable}
The proof of this claim is deferred to \short{the full version}{\Cref{sec:proof_of_gen_witness}}, it generalizes the one of \cite[Lemma 1]{AlonKKLR05}.

\subsection{Generalized Witnesses for \texorpdfstring{$\cP_d$}{Pd}} 
The tester of~\cite{AlonKKLR05}, discussed in \Cref{sec:intro-low-degree},
queries a linear 
$(d+1)$-dimensional subspace; the variant of the tester considered in \cite{BKSSZ10} uses an affine $(d+1)$-dimensional subspace. Both subspaces can be represented as cubes\footnote{We call them cubes for easy visualisation, even though most pattern instances are not parallel to the axes.}. The affine version of the tester is equivalent to $X$-tester with the 
pattern $X$ with $d+1$ parameters $a_1, \dots, a_{d+1}$ representing directions and a translation parameter $a_{d+2}$.
The rows of pattern $X$ are all different binary vectors of $d+2$ bits with value $1$ at the last coordinate. It is easy to verify this pattern satisfies the conditions of~\Cref{claim:low_deg_characterization}.
(For the tester based on a linear subspace simply delete the last column of $X$, which corresponds to fixing the last parameter $a_{d+2} = \vec{0}$.)
For the special case of affinity ($d=1)$, this pattern can be equivalently viewed as taking
parameters $x_1,x_2,x_3$ to testing points $\set{x_1,x_2,x_3,x_1+x_2+x_3}$
(the BLR linearity test simply fixes $x_3=0$), 
and the generalized witness used in Algorithm~\ref{alg:krv-linearity-test} can be viewed as a ``chain of squares'': the second square shares the point $x_1+x_2+x_3$ as if it was its first parameter, and gets additional parameters $x_4, x_5$. When we check the parity of $f$ on all points in both squares, the intermediate point $x_1+x_2+x_3$ appears twice and cancels out. We are left with testing points $x_1, \dots, x_5$ and their sum $\sum_{i\in[5]} x_i$.

To generalize this to higher dimensions, we consider the similar picture for high-dimensional cubes. Alternatively, we can view this as the chain from linearity, but each point is replaced with an entire $(d-1)$-dimensional cube (see \cref{fig:chain_of_cubes}).
Note the parameters of these patterns have two different roles, some parameters impose the chain structure while others are used for the cubical structure.
We dub these patterns ``chains of cubes''.

\begin{definition}[Chain of Cubes Pattern]
    For degree $d$ and odd integer $s$ (length of the chain), we define the {\em chain of cubes pattern} $\chi_{d,s}$.
    It takes parameters $a_1, \dots, a_{d}$ for the cube structure, and $a_{d+1}, \dots, a_{d+s}$ for the chain, and outputs $(s+1) \cdot 2^{d-1}$ testing points, each is a combination of a subset of the first $d$ parameters, and either all or exactly one of the last $s$ parameters. 

    Formally, $\chi_{d,s}$ has ``cube columns'' $1,\dots, d$ and ``chain columns'' $d+1,\dots, d+s$.
    We index each row with $i\in[s+1], j\in \cube{d}$.
    In row $(i,j)$, the first $d$ coordinates are exactly the vector $j$. The last $s$ coordinates depend on $i$, where if $i\in[s]$ then only column $d+i$ has value $1$, and if $i = s+1$ then all $s$ columns have value $1$.
\end{definition}

\begin{claim}
    The chain of cubes pattern $\chi_{d,s}$ is 
    good for $\cP_{d+1}$ for all $d\in \N$ and all odd $s$.
\end{claim}

\begin{proof}
    We show the pattern $\chi_{d,s}$ satisfies the conditions of \Cref{claim:low_deg_characterization}.
    That is, each set of at most $d+1$ coordinates has an even number of rows where all these coordinates are $1$, and there exists a set of $d+2$ coordinates for which the number of such rows is odd. 

    Let $T$ be a set of $d+1$ coordinates. Then $T$ includes at least one chain coordinate, as there are only $d$ cube coordinates.
    First, suppose that $T$
    includes exactly one chain coordinate $c\in\set{d+1,\dots,d+s}$. Rows where all coordinates in $T\setminus\set{c}$ are $1$ come in pairs, indexed either $(c,j)$ or $(s+1,j)$ for any fitting $j$. Thus, there is an even number of them.
    Next, suppose $T$ includes
    at least two coordinates from $\set{d+1,\dots,d+s}$. Then, there exists a cube column $c\in[d]$ such that $c\notin T$. For each row where all coordinates in $T$ are $1$, coordinate $c$ can be 0 or 1, leading to an even number of such rows.

    To show the pattern is minimally sound, consider the set $T = [d+2]$. 
    Two chain coordinates together are $1$ only in rows  
    indexed by $(s+1, j)$ for some $j$. But $j$ must be $\vec 1$ to get $1$ in all the first $d$ coordinates. Thus, only a single row has value $1$ on all coordinates in $T$.  
\end{proof}

\subsection{Soundness of Patterns}
In this section, we show a soundness result for any pattern, and in particular for our chain of cubes pattern, $\chi_{d,s}$.
We resort to two previously known results and combine them together. The first is inherited by the soundness guarantee in~\cite{KS08} for any $2$-ary local characterization.

The other is a specialized argument for affine subsapces that appears in~\cite{BKSSZ10}. The argument deals with small values of $\eps$ (so most violating instances are such due to a single point).
We generalize this argument for any pattern, relying on two properties held by all patterns. First, each two induced testing points, over a random instance, are independent. Second, the tester depends on the parity of all answers to the queried points.

Formally, we prove the following lemma.
\begin{lemma}
    \label{lem:patterns-soundness}
    Fix any pattern $X\in \mathcal{M}_{\ell \times m}(\F_2)$ with $\ell \geq 3$. For any function $f$ that is $\eps$-far from $\cP_d$, we have
    \[
        \Prob{a_1,\dots,a_{m}}{T_{X,f}(a_1,\dots,a_{m}) = 1}
        \geq \min\set{\frac{\ell\eps}{2},\frac{1}{2\ell^2}}.
    \]
\end{lemma}

\begin{proof}
    Our starting point is the soundness guarantee of local characterization, that is Theorem~2.9~in~\cite{KS08}. Any pattern $X\in \mathcal{M}_{\ell \times m}(\F_2)$ is a $2$-ary $\ell$-local characterization, where each row dictates a function, and a vector in $\cube{\ell}$ of evaluations of $f$ is rejected if and only if its hamming weight is odd. Moreover, the first row is linearly independent from all the others (since each two rows are distinct and we are over $\F_2$).
    Thus, for any $\eps$-far function $f$, we get:
    \begin{equation}
        \label{eq:soundness-ks08}
        \Prob{a_1,\dots,a_{m}}{T_{X,f}(a_1,\dots,a_{m}) = 1}
        \geq \min\set{\eps/2, \gamma(\ell)} ,
    \end{equation}  
    where $\gamma(\ell) = 1/((2\ell+1)(\ell-1))$ only depends on $\ell$.

    Next, we adjust the argument from Lemma~8~in~\cite{BKSSZ10}.
    Let $g$ be a degree-$d$ function that achieves minimal distance, $\eps$, from $f$.
    Our probability space is the random choice of parameters $a_1,\dots a_m$. We observe two types of events for each $i\in[\ell]$. The first event is $E_i$ - the event that $f(y_i) \neq g(y_i)$. The second is $F_i$ - the event that $f(y_i) \neq g(y_i)$ but $f(y_j) = g(y_j)$ for all $j\neq i$.

    Since each $y_i$ is uniform over $\F_2^n$, it is clear that $\Prob{}{E_i} = \eps$. Due to the fact that each two rows are different (and hence, over $\F_2$, linearly independent), we get pairwise independence for the variables $y_1,\dots, y_{\ell}$. I.e, for any $i\neq j$ we have $\Prob{}{E_i \wedge E_{j}} = \eps^2$.
    It is clear that for all $i\in [\ell]$:
    \[
        \Prob{}{F_i}
        \geq \Prob{}{E_i} - \sum_{j\neq i} \Prob{}{E_i \wedge E_j}
        \geq \eps - \ell\eps^2 .
    \]

    Note that $g\in\cP_d$ passes the parity-check for every instance, meaning that for $f$ it depends on the number of points in the instance in which it differs from $g$.
    In particular, the test rejects if $f$ differs from $g$ on \emph{exactly one} such point. This event is exactly the disjoint union of all events $F_i$, and its probability is lower bounded by
    \[
        \Prob{}{\bigsqcup_{i\in[\ell]} F_i} 
        = \sum_{i\in[\ell]} \Prob{}{F_i} 
        \geq \ell \left(\eps - \ell\eps^2\right)
        = \ell\eps (1-\ell\eps) .
    \]
     Thus, for any $\eps$-far function $f$, we get
     \begin{equation*}
     \label{eq:soundness-bkssz}
        \Prob{a_1,\dots,a_{m}}{T_{X,f}(a_1,\dots,a_{m}) = 1}
        \geq \ell\eps (1-\ell\eps) . 
     \end{equation*}
 We finish up by combining both arguments. 
    If $\ell\eps\leq 1/2$, 
   we are guaranteed soundness of at least $\ell\eps (1-\ell\eps) \geq \ell\eps/2$.
    Otherwise, $\ell\eps > 1/2 \geq 1/\ell$, which implies $\eps/2 > 1/(2\ell^2)
    $.
    Combined with $\gamma(\ell) \geq 1/(2\ell^2)$, the soundness guaranteed by \Cref{eq:soundness-ks08} must be at least $1/(2\ell^2)$.
\end{proof}

\subsection{Algorithm with Batches of \texorpdfstring{$2^{d-1}$}{exp(d-1)} Queries}
Our low-degree tester is based on multiple simulations of the parity-check tester with the pattern $\chi_{d-1,m}$, but each simulation uses additional queries to overcome the adversary,
mirroring the simulation of $\blr{k}$ in Algorithm~\ref{alg:linearity-main}.
 The number of iterations and the length of the chain, $m$, are chosen according to the soundness guarantee of these generalized patterns\footnote{Improving the soundness guarantee in \Cref{lem:patterns-soundness} to $\Theta\left(\min\set{\ell\eps, 1}\right)$ would allow a better choice $m, r$, similar to Algorithm~\ref{alg:low-degree-main}, and result in a better query complexity.}.

\begin{algorithm}
 \caption{Online-Erasure-Resilient Degree-$d$ Tester}\label{alg:low-degree-main}
{
    \Given {Parameters $\eps\in(0,1/2]$ and $d,t\in\N$; query access to $f$ via 
    $t$-online erasure oracle sequence $\oracle{}$ with batch size $2^{d-1}$}
    $t\gets\max\{t,2\}$
    \Comment{If $t< 2$, replace it with $t=2.$}\\
    $m \gets 2\ceil{\log\left(\frac{2^{20}2^dt^{1/4}  (d\log t)^{3/2}}{\eps^{1/2}}\right)}+1$, $\alpha\gets\min\set{\frac{\eps\ell}{2},  \frac{1}{2\ell^2}}, r \gets \ceil{\frac{2}{\alpha}}$.\\
     \RepTimes{$r$}  {\label{step:low-degree-repeat}
         Sample $V = (v_1,\dots,v_{d-1}) \in (\cube{n})^{d-1}$ uniformly at random.\\
         Sample $X =(x_1,\dots,x_{2m})\in(\cube{n})^{2m}$ uniformly
         at random.\\
         {\bf for each $i\in[2m]$}: batch query $f$ at points $x_i + \sum_{j\in T} v_j$ for all $T \subseteq [d-1]$.\\
         Sample uniformly at random a subset of $[2m]$ of size $m$, denoted $S = \set{s_1,\dots,s_{m}}$.\\
         Set $y = \xor_{j\in S} x_j$ and batch query $f$ at points $y + \sum_{j\in T} v_j$ for all $T \subseteq [d-1]$.\\ 
         \If{$T_{\chi_{d-1,m},\oracle{}}(v_1,\dots,v_{d-1},x_{s_1},\dots,x_{s_m}) = 1$}{\textbf{Reject}
         \Comment{This implies no erasures in this iteration.}
         }
        
    }
        \textbf{Accept}
  }
\end{algorithm}

\subsection{Probability of Seeing an Erasure}

To analyze the probability of seeing an erasure, we use the structured (``cube'') part of our patterns, which is evident in our tester as well. This is done, in essence, by considering the subspace $A = sp\set{v_1,\dots,v_{d-1}}$ and observing the quotient group $\F_2^n / A$, where each element has the form $x + A$ (corresponding to $2^{d-1}$ points in $\F_2^n$). Projecting all our queries and erasures to this group brings us back to the linearity tester.  

More formally, assume $dim(A) = d-1$ (if it is smaller, the argument only improves), and let $U = (u_d, \dots, u_n)$ complete the set of vectors $V$ to a basis of the entire space.
Define $B = sp\set{u_d,\dots,u_n}$, and let $\pi_B$ be a projections of any vector $x\in\F_2^n$ to $B$ using its unique representation over the basis $V \cup U$.
The image of the projection is $B$, which is isomorphic to $\F_2^{n-d+1}$, and for each $u\in B$, its pre-image is exactly $\pi_B^{-1}(u) = u + A = \set{u+v \ : \  v\in A}$.

We can therefore recast a single iteration of Algorithm~\ref{alg:low-degree-main} in terms of $B$ alone. Projection of a uniformly random point $x_i\in \F_2^n$ is a uniformly random point $u_i\in B$, with $x_i + A = u_i + A$ (which defines the same batch of queries).
Furthermore, after choosing the random subset of indices $S$, we get
\[
    \pi_B(y) = \pi_B\left(\xor_{j\in S} x_j\right) = \xor_{j\in S} \pi_B(x_j) \in B ,
\]
for which we again query $y + A = \pi_B(y) + A$.
It is now clear that an iteration of Algorithm~\ref{alg:low-degree-main} reduced to its queries on $B$ is equivalent to an iteration of Algorithm~\ref{alg:linearity-main} on an entire space of dimension $n' = n - d +1$ (which $B$ is isomorphic to).
Lastly, we strengthen our adversary, by saying each time it tries to erase a point $x\in \F_2^n$, the entire set $x + A$ is erased instead, and we think of it as an erasure of $\pi_B(x)$ in $B$. Thus, the probability of seeing an erasure, even with the strengthened adversary, is that of Algorithm~\ref{alg:linearity-main} seeing an erasure.

Applying the same arguments as in \Cref{claim:no_collisions_for_y}, with $n' = n - d + 1$, the adjusted number of batches, $2m$, and adjusted premise on $t$ leads to the following claim.
\begin{claim}
    \label{claim:low-degree_no_collisions_for_y}
    Fix an iteration of the loop in Step~\ref{step:low-degree-repeat} of Algorithm~\ref{alg:low-degree-main}, and denote $u_i = \pi_B(x_i)$. 
    For all $T\subseteq[2m]$,
    let $y_{_T} = \bxor_{j\in T} u_j$.
    Then, the probability there exist two subsets $T_1\neq T_2$ of the set $[2m]$ with $y_{_{T_1}} = y_{_{T_2}}$ is small:
    \[
        \Prob{X}{\exists T_1\neq T_2 \text{\rm\ such that\ } y_{_{T_1}} = y_{_{T_2}}} 
        \leq\frac{2^{4m}}{2^{n'}}
        \leq \frac{\eps}{2^{20} 2^{2d} (d\log t)^2},
    \]
    as long as $t \log^7 t\leq c\cdot \eps^{2.5}d^{-7} 2^{(n-11d)/2}$ for some constant $c$.
\end{claim}

Finally, we can obtain the following lemma. We omit its proof which is identical to the proof of \Cref{claim:linearity_no_erasures}, using $n'\geq 4m$ (instead of $n\ge 2m$ there).

\begin{lemma}
    \label{claim:lowd-egree_no_erasures}
    For all $m\geq 5$, and $t$ as above, the probability that one specific iteration  of the loop in Line~\ref{step:low-degree-repeat} of Algorithm~\ref{alg:low-degree-main} queries an erased point is at most
    $$\frac{4tr m^2}{2^{2m+1}} + \frac{\eps}{2^{20} 2^{2d} (d\log t)^2} .$$
\end{lemma}

\subsection{Analysis of the Online Tester}

Now we complete the analysis of Algorithm~\ref{alg:low-degree-main} and prove \Cref{thm:d_deg:chain_of_cubes} 
for erasure adversaries. 
Similar to the case of linearity, our tester has a small constant probability of ever seeing a manipulated entry, and hence it is robust against corruptions as well.
\short{}{
\dDegChainOfCubesThm*}
\begin{proof}[\short{Proof of Theorem~\ref{thm:d_deg:chain_of_cubes} for the case of erasures}{Proof for the case of erasures}]
We show that Algorithm~\ref{alg:low-degree-main} satisfies the conditions of the theorem. It always accepts all degree-$d$ functions.
Now, fix an adversarial (budget-managing) strategy and suppose that the input function is $\eps$-far from any degree $d$ function.

By \Cref{lem:patterns-soundness} with $\ell = (m+1)2^{d-1}$ points in the $\chi_{d-1,m}$ pattern, the probability that one iteration the loop in Step~\ref{step:low-degree-repeat}  samples a witness for not being degree $d$ is 
$\alpha=\min\set{\frac{\ell\eps}{2},  \frac{1}{2\ell^2}}.$ 
We split the analysis into two cases, depending on which term achieves the minimum.

        \textbf{Case I:} 
        $\alpha=\ell\eps/2$,
    the number of iterations is $r=\ceil{\frac{2}{\alpha}}=\ceil{\frac{4}{\ell\eps}}$
    which means that $r\ell\eps \leq \left(\frac{4}{\ell\eps}+1\right)\ell\eps=4+\ell\eps\leq 5$. 
    By \Cref{claim:lowd-egree_no_erasures}, the probability that an erasure is seen within a single iteration is at most
  \begin{align*}\label{eq:low-degree-bound-cases}
         \frac{4tr m^2}{2^{2m+1}} + \frac{\eps}{2^{20} 2^{2d} (d\log t)^2} .
     \end{align*}
    We bound the second term naively by $\frac{\eps}{100}$. For the first term, we use the setting of $m$ in the algorithm which implies $2^{2m}\geq \frac{2^{20}t}{\eps^2}$, and the fact that $m \leq \ell$ to get
    \[
   \frac {4r m^2 \eps^2}{2^{20}}\leq \frac {4(r m \eps) m\eps }{2^{20}} 
    \leq
    \frac {20 m \eps}{2^{20}} 
    \leq \frac { m \eps}{2^{15}}.
    \]
    The probability that an erasure is seen in a specific iteration is therefore at most
    \begin{align*}\label{eq:low-degree-case1}
       \frac { m \eps}{2^{15}} + \frac {\eps}{100}
      < \frac {\ell \eps}{80}.
    \end{align*}
     By a union bound, the probability of a single iteration seeing an erasure or not selecting a witness is at most $1-\frac{\ell\eps}{2} + \frac{\ell\eps}{80} \leq 1 - \frac{\ell\eps}{3}$.
    Algorithm~\ref{alg:low-degree-main} errs only if this occurs in all iterations. By independence of random choices in different iterations, the failure probability is at most
     \[
        \Big(1-\frac{\ell\eps}{3}\Big)^r
        \leq \Big(1-\frac{\ell\eps}{3}\Big)^{\frac{4}{\ell\eps}}
        \leq e^{-\frac 4 3} 
        \leq \frac{1}{3},
    \]
where we used $1-x\leq e^{-x}$ for all $x$. 
\\
The total number of queries used in this case is at most $r(2m+1)b \leq 2r\ell \leq 10/\eps$.
    
    \textbf{Case II:}
    $\alpha=\frac{1}{2\ell^2}$, and $r = 4\ell^2$. We first bound $m$, set in Algorithm~\ref{alg:low-degree-main}, in terms of $d$ and $\log t$.
        \begin{align*}
           m 
          &\leq 41 + 2d+\log t +3\log d+ 3\log\log t + \log(1/\eps)\\
          &\leq 41+4\log t+8d+3\log m\\
          &\leq 20(\log t+d)+3\log m,
        \end{align*}
        The second inequality follows from the fact that, in this case, $\log(1/\eps)\leq 3\log \ell\leq 3(\log m+d)$.
        Since $m\geq 2^5$ we have $3\log m\leq m/2$ which implies (for $t,d \geq 2$)
        $$m\leq 40(\log t+d)\leq 2^6 d \log t .$$  
        To bound the probability that an erasure is seen in one iteration, we start, like in Case I, with \Cref{claim:lowd-egree_no_erasures}.
        We bound the first term recalling that $r=4\ell^2\leq 4m^2 2^{2d}$,
        \[
    \frac{4trm^2}{2^{2m+1}}\leq \frac {8t m^42^{2d} }{2^{2m}}
     \leq \frac{2^3 \cdot 2^{2d} m^4\eps^2}{2^{80}2^{4d}(d\log t)^6}
     \leq \frac{1}{2^{40} m^2 2^{2d}}
     \leq \frac{1}{20\ell^2} ,
\]
where the second inequality plugs in the value of $m$ set in the algorithm to bound $t/2^{2m}$, and the third uses $m\leq 40 d \log t$ in the denominator.
Adding the second term, we get at most
\[ 
    \frac{4tr m^2}{2^{2m+1}} + \frac{\eps}{2^{20}2^{2d} (d \log t)^2}
    \leq \frac{1}{20\ell^2} + \frac{1}{20\ell^2} \leq \frac{1}{10\ell^2}.
\]
   
    By a union bound, the probability of a single iteration seeing an erasure or not selecting a witness is at most $1- \frac{1}{2\ell^2} + \frac 1{10\ell^2} \leq 1-\frac{1}{3\ell^2}$. As before, we use the independence of different iterations to conclude the algorithm errs with probability at most 
   $\big(1-\frac 1 {3\ell^2}\big)^{4\ell^2}
   \leq \frac{1}{3}.$
   \\
    The total number of queries used is
    $r(2m+1)b \leq 6m^3 2^{3d} \leq 240\cdot 2^{3d}(d+\log t)^3$.
\end{proof}

\section{Online Erasure-Resilient Testing of Local Sequential Properties}\label{sec:local}

In this section, we establish our results on testing local 
properties of sequences discussed in~\Cref{sec:intro-local-properties}. That is, we prove the following two main results.
\TestingLocalPropertiesBudgetManagingTHM*
\SortednessBatchTwoTHM*

Our core technical result in this section shows that a simple pair tester (presented in Algorithm~\ref{alg:pair_tester_local_properties}) can test all local properties of sequences. 

For the algorithm, for $d, n \in \N$, where $d<n$, 
we define the set of all distance-$d$ pairs in $[n]$ by\footnote{Note that the definition is cyclic, in the sense that allows for pairs where $x$ is close to $n$, and $y$ is much smaller and close to $1$. The cyclic nature of the definition slightly simplifies the algorithm and analysis.}
$$
D_d(n) = \{(x,y) \in [n]^2 : y-x \equiv d \ (\text{mod}\ n)\}.
$$
We also define the \emph{interval} $ 
[a:b]$ as the set of all integers $\{i \in \mathbb{Z} : a \leq i \leq b\}$. The \emph{length} of the interval $[a:b]$ is $b-a+1$. Our algorithm relies on a fundamental notion of \emph{unrepairability} in a sequence $f$ with respect to a property $\mathcal{P}$. The notion we use  
is a slightly generalized version of a notion of unrepairability for intervals originally proposed in \cite{BE19}.
\begin{definition}[Unrepairability]
Let $\mathcal{P}$ be a local property of sequences $f \colon [n] \to \R$ characterized by the forbidden family $\mathcal{F}$. 
For a subset $S \subseteq [n]$, a sequence $f$ is \emph{unrepairable} on $S$ with respect to $\mathcal{P}$ if the following holds: 
every sequence $f' \colon [n] \to \R$, where $f'(x) = f(x)$ for all $x \in S$,  
does not satisfy $\mathcal{P}$. (I.e., every $f'$ that agrees with $f$ on  
all entries in $S$ contains a pattern from $\mathcal{F}$.)
\end{definition}

Crucially, given $\mathcal{P}$ and query access to $f$, in order to know whether $f$ is unrepairable on $S$ w.r.t.$~\mathcal{P}$, one needs to query $f$ only on the elements of $S$. For our purposes, $|S| \leq 2$ always holds, and so  
we only need to make two queries in order to check unrepairability.

Note that, up to some edge cases, an interval which satisfies the notion of unrepairability from \cite{BE19} is called 
a \emph{witness interval} in this paper (see \Cref{def:interval_unrepairability} below); the notion of a witness interval is used in the analysis, but not in the algorithm.

\begin{algorithm}[]
 \caption{Pair tester for local properties}\label{alg:pair_tester_local_properties}
    \Given {local property $\mathcal{P}$, $\eps\in(0,1)$; 
    query access to $f : [n] \to \R$ (in offline model)  
    \\
    or, for $t \in \R$ and $b \in \N$,  via $t$-online erasure oracle sequence $\oracle{}$ with batch size $b$} 
    Set $\ell = \lfloor \log \left(\frac{\eps n}{4}\right)\rfloor$.\\
    \RepTimes{$\frac{200 \log (\eps n)}{\eps}$}{
    Sample $i \in [0 : \ell]$ uniformly at random,\\
    Sample $(x,y) \in D_{2^i}(n)$ uniformly at random;
    \textbf{query} $f(x)$ and $f(y)$.\\    
    \If{$f$ is unrepairable on $\{x,y\}$ 
    }{\textbf{Reject}}
    }
    \textbf{Accept}
 \end{algorithm}
 
We first show that the pair tester works for every local property $\mathcal{P}$ in the \emph{offline} property testing model. The statement is given in the following lemma and proved in
\Cref{sec:pair_tester_via_shifted}.
\begin{lemma}[Offline correctness of pair tester]\label{lem:pair_tester_offline}
Algorithm~\ref{alg:pair_tester_local_properties} is a nonadaptive $\eps$-tester with one-sided error probability bounded by $1/7$ for  
every local
property $\mathcal{P}$ in the offline testing model (without erasures or corruptions).
\end{lemma}

The pair tester not only works in the offline setting, but it is also unlikely to see erasures in the online setting when $t$ is appropriately small. The next two lemmas state results of this type for batch sizes 1 and 2,  
respectively.

\begin{lemma}[Pair tester, batch size 1]
\label{lem:local_batch_one}
Fix a local property $\mathcal{P}$, $\eps > 0$, $n \in \N$, and $t \leq \eps / 10^6$. The probability that
Algorithm \ref{alg:pair_tester_local_properties}, when run in the online erasure model with parameters $\mathcal{P}, \eps, n, t$ and batch size $b=1$, queries an erased element is at most $1/7$. \end{lemma}

\begin{lemma}[Pair tester, batch size 2]
 \label{lem:local_batch_two}
Fix a local property $\mathcal{P}$, $\eps > 0$, $n \in \N$, and $t \leq \frac{ \eps^2 n}{3  \cdot 10^5 \log^{2} (\eps n)}$. The probability that
Algorithm \ref{alg:pair_tester_local_properties}, when run in the online erasure model with parameters $\mathcal{P}, \eps, n, t$ and batch size $b=2$, queries an erased element is at most $ 1/7$. \end{lemma}

Note the sharp contrast between the threshold rate for $b=1$ ($t = \Theta(\eps)$) and $b=2$ ($t = \tilde{\Theta}(n)$). 

The positive results of Theorems~\ref{thm:sortedness-budget-managing} and~\ref{thm:sortedness_batch_2}  
follow from Lemmas~\ref{lem:pair_tester_offline}, \ref{lem:local_batch_one}, \ref{lem:local_batch_two}. The proof of the negative (impossibility) result of~\Cref{thm:sortedness-budget-managing} is given in~\Cref{sec:impossibility_sortedness}.

\begin{proof}[Proof of~\Cref{thm:sortedness-budget-managing}, Part 1]
For an erasure adversary, 
we use
Lemmas~\ref{lem:pair_tester_offline} and~\ref{lem:local_batch_one}. 
Consider a local property $\cP$ and a sequence $f$ that is $\eps$-far from $\cP$.
Let $A$ be the event that the tester rejects, 
given query access to $f$ 
(not to the erasure oracle). Let $B$ be the event that the tester encounters at least one erasure. The rejection probability of the tester in the online setting is at least 
$$
\Prob{}{A \setminus B} \geq \Prob{}{A} - \Prob{}{B} \geq \frac{6}{7}-\frac{1}{7} > \frac{2}{3},
$$
where the second inequality follows from Lemmas~\ref{lem:pair_tester_offline} and~\ref{lem:local_batch_one}.

For a corruption adversary, 
we use the same argument combined with \cite[Lemma 1.8]{KalemajRV22}. 
\end{proof}

\begin{proof}[Proof of~\Cref{thm:sortedness_batch_2}]
The proof (for both erasures and corruptions) is identical to the above proof of~\Cref{thm:sortedness-budget-managing}, Part 1, except that we use~\Cref{lem:local_batch_two} instead of~\Cref{lem:local_batch_one}.
\end{proof}

We next state our main result for \emph{fixed-rate} adversaries, showing that the threshold rate is arbitrarily close to $1$. The proof appears in~\Cref{sec:pair_tester_erasure_resilience}. A matching negative result is established in \cite{KalemajRV22}, see Theorems 1.5 and 1.7 there. 
\begin{prop}\label{prop:local-fixed-rate}
Fix $\eps > 0$, $t < 1$ and a local property $\mathcal{P}$ of sequences $f \colon [n] \to \R$, where $n \geq \frac{C\left(\log \frac 1\eps + \log \frac {1}{1-t}\right)^2}{\eps^2 (1-t)}$ for a large enough constant $C > 0$. There exists a (one-sided error) 
$\eps$-tester for $\mathcal{P}$  
that works in the presence of a $t$-online erasure fixed-rate adversary and has 
query complexity  
$O\left(\frac{\log (\eps n)}{\eps (1-t)}\right)$. For a corruption adversary, the same is true without the one-sided error guarantee.
\end{prop}

\subsection{Online Erasure-Resilience of Pair Tester}\label{sec:pair_tester_erasure_resilience}

We next prove
that the pair tester is unlikely to see an erased element during the course of the algorithm. In the following definition, we separate the erasures into two types.
\begin{definition}[Blind 
and relative erasures]
Consider a single loop iteration of Algorithm~\ref{alg:pair_tester_local_properties}. An 
erasure made 
before the first query in this iteration is called \emph{blind} (with respect to this iteration). An erasure made after the first query but before the second query of the iteration is called \emph{relative}. 
\end{definition}

Relative erasures can be much more effective than blind ones in making the tester query an erased element.
For relative erasures, observe that once the first query, $x$, in a given iteration is made, the second query can take one of roughly $\log (\eps n)$ options, and its probability to encounter a relative erasure made during the iteration is of order $1 / \log (\eps n)$. In contrast,~\Cref{lem:marginal_pair_tester} shows that the probability of any future query to see a specific blind erasure is much smaller: $O(1/n)$.

\begin{lemma}[Pair tester: marginal probability]\label{lem:marginal_pair_tester}
Consider one iteration of the loop in Algorithm~\ref{alg:pair_tester_local_properties} with parameters $n, \mathcal{P}, \eps, t$. 
Fix $z \in [n]$. The probability that $z$ is queried in this iteration is  $2/n$.
\end{lemma}
\begin{proof}
Fix the value of $i$ sampled in Line 3 of Algorithm \ref{alg:pair_tester_local_properties}. Out of the $n$ pairs in $D_{2^i}(n)$, exactly two contain $z$, so the probability of $z$ to be queried is exactly $2/n$.
\end{proof}

This sharp contrast between relative and blind erasures is the reason for the huge separation between the batch-$1$ case and the batch-$2$ case. In the former case, both relative and blind erasures are available to the adversary, whereas in the latter case, only blind erasures are available.

The following simple corollary bounds the probability that the pair tester queries
at least one blind erasure in the course of the algorithm.
\begin{corollary}\label{coro:marginal_local}
Consider one run of Algorithm \ref{alg:pair_tester_local_properties}. Let $q$ be the number of queries made by the tester, and let $T = tq$ denote the erasure budget for the run. The probability that the tester encounters a blind erasure is at most $\frac{2Tq}{n}$.
\end{corollary}
\begin{proof}
Consider a query $x$ made by the tester. Let $S$ be the set of blind erasures (with respect to the current iteration) made so far. The randomness used by Algorithm \ref{alg:pair_tester_local_properties} to choose $x$ is independent of the set $S$ (by the structure of the tester and since $S$ consists only of blind erasures). By~\Cref{lem:marginal_pair_tester},
for each query we have $\Prob{}{x \in S} \leq 2\card{S} / n \leq 2T/n$.
The corollary follows by a union bound over all queries.
\end{proof}

We now prove that our testers are unlikely to encounter erasures for both batch sizes.

\begin{proof}[Proof of~\Cref{lem:local_batch_one}]
The number of queries made by Algorithm~\ref{alg:pair_tester_local_properties} is $q = \frac{200 \log (\eps n)}{\eps}$. The total number of erasures is at most $T=tq \leq \frac{2}{10^4} \log (\eps n)$. 
By \Cref{coro:marginal_local},
the probability that the algorithm queries a blind erasure at least once during the course of its run is bounded by
$$
\frac{2Tq}{n} \leq \frac{800 \log^{2} \eps n}{10^4 \cdot n} < \frac{1}{10}
$$ for all $n$.

Next we bound the probability that Algorithm~\ref{alg:pair_tester_local_properties} encounters a relative erasure. The a priori probability
that a viable pair $(x,y)$ (that is, any pair $(x,y) \in D_{2^i}(n)$ for some choice of $i \in [0:\ell]$)
is queried in a given 
iteration is exactly $\frac{1}{n(\ell+1)}$, where $\ell$ is as in the algorithm. Note that this probability does not depend on $i$. By Bayes' rule, conditioned on the first element being $x$, the probability of each of the 
elements in the set
\[
    \set{y \in [n] : \exists  i\in [0:\ell] \text{ such that } (x,y) \in D_{2^i}(n)}
\]
to be the second query of the iteration is exactly $\frac{1}{\ell + 1} \leq \frac{2}{\log (\eps n)}$; all other values of $y$ have probability zero to be queried. Therefore, for any specific relative erasure conducted by the adversary, the probability that the second query of the relevant iteration hits the erasure is at most $\frac{2}{\log (\eps n)}$. By a union bound (noting that, by definition, each relative erasure can only affect a single iteration of the algorithm), the probability that the algorithm queries a relative erasure throughout its run is at most $\frac{2}{\log (\eps n)} \cdot T \leq \frac{4}{10^4}$. The proof follows by combining the bounds for blind and relative erasures.
\end{proof}

\begin{proof}[Proof of~\Cref{lem:local_batch_two}]
Set $C = 1/(3\cdot 10^5)$.
In the batch size $2$ case there are no relative erasures, since each loop iteration in the algorithm is captured by a single batch. That is, all erasures are blind.
We apply~\Cref{coro:marginal_local} with 
$$
q = \frac{200 \log (\eps n)}{\eps} \qquad \text{and} \qquad
T = tq \leq \frac{C \eps^2 n}{\log^{2} (\eps n)} \cdot \frac{200 \log (\eps n)}{\eps} = \frac{200C \eps n}{\log (\eps n)}
$$
to bound the probability to query at least one (blind) erasure during the run of the algorithm by
$$
\frac{2Tq}{n} \leq \frac{200C\eps}{\log (\eps n)} \cdot \frac{200 \log (\eps n)}{\eps} \leq 40,000 C < \frac{1}{7},
$$
as desired.
\end{proof}

\begin{proof}[Proof of~\Cref{prop:local-fixed-rate}]
Our tester applies the pair tester (Algorithm \ref{alg:pair_tester_local_properties}) as a black box, but only at prespecified time stamps along the way (in which the fixed-rate adversary is known not to make erasures). Let $m = O(\eps^{-1} \log (\eps n))$ denote the number of iterations in Algorithm \ref{alg:pair_tester_local_properties}. 
By definition of the fixed-rate model (\Cref{def:fixed-rate_budget-managing}), when the rate is $t < 1$, there exists a sequence of integers $i_1 < i_2 < \ldots < i_m$, which satisfies the following properties.
\begin{itemize}
\item For each $j \in [m]$, the adversary cannot make an
erasure after the query at time $i_j$ is conducted.
\item For each $j \in [m-1]$, we have $i_{j+1} - i_j \leq \frac 2{1-t}$.
\end{itemize}

Our tester in this case simulates Algorithm \ref{alg:pair_tester_local_properties}:  for each $j=1,\ldots,m$ we run iteration $j$ of the latter at timestamps $i_j$ and $i_j+1$. Note that there is no erasure between them (so all erasures are blind).
 Let $q \leq 2m \cdot \frac{2}{1-t} = O\left(\frac{\log (\eps n)}{\eps(1-t)}\right)$ be the total number of queries in our case.
By~\Cref{lem:marginal_pair_tester}, the probability that the query at each time $i_j$ or $i_j+1$ encounters an erased element is at most $2q/n$. By a union bound, the probability that we observe at least one erasure along the way is bounded by
\begin{equation}
\label{eqn:erasure_fixed_rate_local}
\frac{2q}{n} \cdot m \leq \frac{1}{100},
\end{equation}
provided that $C$ is large enough.
Now, disregarding erasures, by~\Cref{lem:pair_tester_offline} if $f$ is $\eps$-far from $\mathcal{P}$ then the probability that the tester queries a violation to $\mathcal{P}$ is at least $6/7$. Combined with \eqref{eqn:erasure_fixed_rate_local}, the proof follows for the online erasures setting. The case of online corruptions immediately follows from~\cite[Lemma 1.8]{KalemajRV22}.
\end{proof}

\subsection{Pair Tester for Local Properties via Shifted Interval Partitioning}
\label{sec:pair_tester_via_shifted}

The main purpose of this section is to prove~\Cref{lem:pair_tester_offline}, establishing the correctness of the pair tester in the offline setting. 
We first present and analyze Algorithm~\ref{alg:shifted_balanced_tree_tester_local_properties}, a randomly shifted version of the structured tester from~\cite{BE19} for local properties. We then show that the pair tester can simulate the shifted structured tester.

The next two definitions present a hierarchical partition structure of intervals (with some overlaps between endpoints of intervals), used in Algorithm~\ref{alg:shifted_balanced_tree_tester_local_properties}.

\begin{definition}[Shifted hierarchical partition]
Let $a, n, w \in [n]$ and let $\ell \geq 0$ be an integer, where $w$ is divisible by $2^\ell$ and 
$a + w \leq n$. An \emph{$a$-shifted hierarchical partition of width $w$ with $\ell$ layers} in $[n]$,
or
\emph{$a$-shifted $(\ell, w)$-hierarchical partition} in short, is a tuple $(T_0, \ldots, T_\ell)$ where for each $i \in [0:\ell]$, the $i$-th layer $T_i$ is a set of intervals defined by 
$$
T_i = \{[a+j\cdot 2^i : a+(j+1)\cdot 2^i]\ \ | \ \ j \in [w/2^i-2]\} \quad\bigcup\quad \left\{ [1 : a+2^i], [a+w-2^i : n]\right\}.
$$
\end{definition}
In other words, each layer $T_i$ consists of intervals of length $2^i + 1$, except the first interval (starting at 1) and the last interval (ending at $n$), which may be longer. The intervals in each layer $T_i$ intersect only at their endpoints.

The first and last intervals in each layer $T_i$ play a special role in our analysis.
\begin{definition}[Extremal intervals]
Let $T$ be an $a$-shifted $(\ell, w)$-hierarchical partition of $[n]$. An interval $I = [x:y] \in T_i$ is said to be \emph{left-extremal} if $x=1$ and \emph{right-extremal} if $y=n$. In either of these cases, $I$ is considered \emph{extremal}.

The \emph{query-pair} of an interval $I$ with respect to $T$ and $i$, denoted $Q_{T,i}(I)$, is defined as $\{x,y\}$ if $I$ is not extremal, $\{a, a+2^i\}$ if $I$ is left-extremal, and $\{a+w-2^i, a+w\}$ if $I$ is right-extremal.
\end{definition}

Equipped with a shifted hierarchical partition structure, we now present the algorithm.

\begin{algorithm}[]
 \caption{Shifted hierarchical structured tester for local properties in offline model}\label{alg:shifted_balanced_tree_tester_local_properties}
    \Given {local property $\mathcal{P}$, $\eps\in(0,1)$; query access to $f : [n] \to \R$ (in offline model)
    }
    Set $\ell = \lfloor \log \left(\frac{\eps n}{4}\right)\rfloor$.\\
    Let $w \in \N$ be the maximum integer divisible by $2^\ell$ and satisfying $w \leq n -  \eps n / 4 $.\\
    \RepTimes{$\frac{20 \log (\eps n)}{\eps}$}{
    Sample $a \in [\lceil\eps n / 4\rceil]$ uniformly at random. \\
    Let $T = (T_0, \ldots, T_\ell)$ be the $a$-shifted $(\ell,w)$-hierarchical partition of $[n]$. \\
    Sample a uniformly random $i \in [0 : \ell]$
    and a uniformly random interval $[x,y] \in T_i$ \label{step:sample-level-interval}.\\
    \textbf{Query} $f(x')$ and $f(y')$, where $\{x',y'\} = Q_{T,i}([x,y])$. \label{step:query-pair}\\
    \If{$f$ is unrepairable on $\{x',y'\}$}{\textbf{Reject}
    }}
    \textbf{Accept}
 \end{algorithm}

The following lemma establishes that Algorithm~\ref{alg:shifted_balanced_tree_tester_local_properties} is an $\eps$-tester for local properties in the standard \emph{offline} model of property testing.
\begin{lemma}\label{lem:balanced_tree} Let $n \in \N$ and $\mathcal{P}$ be a local property of sequences $f \colon [n] \to \R$.
Algorithm~\ref{alg:shifted_balanced_tree_tester_local_properties} is an $\eps$-tester for $\mathcal{P}$ with one-sided error  and success probability $6/7$ in the offline model (with no erasures/corruptions).
\end{lemma}

For the proof, we define the following notion of a witness interval. 

\begin{definition}[Witness interval]
\label{def:interval_unrepairability}

Let $\mathcal{P}$ be a local property of length-$n$ sequences, and $\mathcal{F}$ be the family of forbidden patterns characterizing $\mathcal{P}$. The interval $I=[a:b]$ is said to be a \emph{witness interval} for the sequence $f \colon [n] \to \R$ not satisfying $\mathcal{P}$, if one of the following holds:
\begin{itemize}
\item $a > 1$ and $b < n$, and all sequences $f' \colon [b-a+1] \to \R$, where $f'(1) = f(a)$ and $f'(b-a+1) = f(b)$, contain at least one copy of a (forbidden) pattern from $\mathcal{F}$.
\item $a=1$, and all sequences $f' \colon [b] \to \R$ with $f'(b) = f(b)$ contain a pattern from $\mathcal{F}$. 
\item $b=n$, and all $f' \colon [n-a+1] \to \R$ with $f'(1) = f(a)$ contain a pattern from $\mathcal{F}$.
\end{itemize}

Let $T = (T_0, T_1, \ldots, T_\ell)$ be an $a$-shifted $(\ell,w)$-hierarchical partition of $[n]$. 
For all $i,j\in[0:\ell]$, an interval $J\in T_j$ is an \emph{ancestor} of an interval  $I\in T_i$ if $i<j$ and $I\subset J$.
For $i < \ell$, the \emph{parent} of $I$ is its unique ancestor in $T_{i+1}$. An interval $I$ in $\bigcup_{i=0}^{\ell} T_i$ is a \emph{maximal witness} (with respect to $f, \mathcal{P}, T$) if $I$ is a witness interval for $f$ with respect to $\mathcal{P}$, while all ancestors of $I$ are not witnesses.
\end{definition}

To know whether $[a:b]$ is a witness interval for $f$ with respect to a given property, one only needs to query $f(a)$ and $f(b)$ in the case that $a>1$ and $b < n$; only $f(b)$ when $a=1$; and only $f(a)$ when $b=n$.
 The main structural statement connecting being far from a property $\mathcal{P}$ and the structure of witness intervals for $f$ in the hierarchical partition is as follows.

\begin{lemma}
\label{lem:structural_charac_unrepairability}
Let $f \colon [n] \to \R$ be $\eps$-far from a local property $\mathcal{P}$. Let $T = (T_0, \ldots, T_\ell)$ be an $a$-shifted $(\ell,w)$-hierarchical partition of $[n]$, where $a \leq \lceil \eps n / 4 \rceil$, and $\ell$ and $w$ are as defined in Algorithm \ref{alg:shifted_balanced_tree_tester_local_properties}.
The set $\mathcal{U}$ of all maximal witness intervals in $f$ with respect to $T$ and $\mathcal{P}$ satisfies:\footnote{Note that $T_i$ and $T_i \cap \mathcal{U}$ are collections of sets; for example, $|T_i|$ is the \emph{number of intervals} in $T_i$ (and not, say, the total number of elements in all intervals in $T_i$).}
$$\sum_{i=0}^{\ell}  \frac{|T_i \cap \mathcal{U}|}{|T_i|} \geq \frac{\eps}{8}.$$
\end{lemma}
\begin{proof}
If $T_\ell \cap \mathcal{U} \neq \emptyset$, the lemma follows since $|T_\ell| \leq \frac{n}{\eps n / 8} = \frac 8 \eps$. Otherwise, each $I \in \mathcal{U}$ has a (non-witness) parent $P(I)$ in $T$. For convenience, for each interval $I$, define $g(I) := I \cap [a:a+w]$.

We first claim that $|g(P(I))| = 2|g(I)|-1$. Indeed, for $I \in T_i$ for some $i \in [0:\ell]$, we have
\begin{equation}
\label{eq:g_I}
|g(I)| = 2^i+1 \quad \text{and} \quad |g(P(I))| = 2^{i+1} + 1 = 2|g(I)|-1.
\end{equation}
Note that the above is trivial for non-extremal intervals, and follows from definition for the extremal ones. Consequently,
\begin{equation}
\label{eqn:parents_half}
1+\sum_{I \in \mathcal{U}} |g(P(I))| \leq 2\sum_{I \in \mathcal{U}} |g(I)|.
\end{equation}

We now bound the left-hand side from below, assuming that $f$ is $\eps$-far from $\mathcal{P}$. 
For any interval $I = [b:c]$, define its set of endpoints as $\{b,c\} \setminus \{1,n\}$, and define the interior $\Int(I)$ of $I$ as the set of all non-endpoint elements in $I$.
 Let
$$S = \bigcup_{I \in \mathcal{U}} \Int(P(I)) \ \cup \ [1:a] \ \cup\  [a+w:n].$$ 
\begin{claim}\label{claim:modify_S}
There exists $f' \in \mathcal{P}$ such that $f'(x) = f(x)$ for all $x \in [n] \setminus S$.
\end{claim}
\begin{proof} 
Let $\mathcal{F}$ be the forbidden family characterizing $\mathcal{P}$.

First, for any pair of consecutive elements $j,j+1$ which both belong to $[n] \setminus S$, consider the interval $A_j = [j:j+1] \in T_0$. This interval is not contained in an interval from $\mathcal{U}$. By the maximality of $\mathcal{U}$, the interval $A_j$ must be a non-witness, meaning that $(f(j),f(j+1))\notin \mathcal{F}$.

Let $L$ denote the minimal interval in $T$ which contains $[1:a+1]$ and is a non-witness. Note that either $L=[1:a+1]$ or $L = P(I)$ for some $I \in \mathcal{U}$. Define $R$ similarly replacing $[1:a+1]$ with $[a+w-1:n]$. The intervals $L,R$ and all intervals $P(I)$, for $I \in \mathcal{U}$, are non-witnesses (by definition of $\mathcal{U}$).
Thus, one can remove all $\mathcal{F}$-copies from these intervals by only modifying their interiors. 

Let $f'$ be the sequence resulting from $f$ by making these modifications (only in interiors of intervals from $\{L,R\} \cup \{P(I) : I \in \mathcal{U}\}$; for all other elements $x$, we set $f'(x) = f(x)$). Observe that $f'(x) = f(x)$ for $x \notin S$. Indeed, since $\Int(L) \subseteq [1:a] \cup \bigcup_{I \in \mathcal{U}} \Int(P(I))$ and $\Int(R) \subseteq [a+w:n] \cup \bigcup_{I \in \mathcal{U}} \Int(P(I))$, the modified elements all lie in $S$.

Our choice of $f'$ implies that $(f'(j),f'(j+1)) \notin \mathcal{F}$ for each pair of elements $\{j,j+1\}$ with nonempty intersection with $S$. 
For pairs $\{j,j+1\} \subseteq n \setminus {S}$, we have seen in the first paragraph that $(f(j), f(j+1)) \notin \mathcal{F}$. Since $j,j+1 \notin S$, we have $f'(j) = f(j)$ and $f'(j+1) = f(j+1)$. Hence, $(f'(j), f'(j+1)) \notin \mathcal{F}$. Thus, $f'$ does not contain $\mathcal{F}$-copies,
completing the proof of \Cref{claim:modify_S}.
\end{proof}

We now complete the proof of~\Cref{lem:structural_charac_unrepairability}. Since $f$ is $\eps$-far from $\mathcal{P}$, we have that $|S| \geq \eps n$. On the other hand, $|[1:a]| + |[a+w:n]| = n-w+1 \leq \frac{\eps n}{4} + 2^\ell + 1 \leq \frac{\eps n}{2} + 1$ and so $ \bigcup_{I \in \mathcal{U}} |g(P(I))| \geq \eps n - \frac{\eps n}{2} - 1 = \frac{\eps n}{2} - 1$. From \eqref{eqn:parents_half}, we conclude that $\sum_{I \in \mathcal{U}} |g(I)| \geq \frac{\eps n}{4}$.

Since $|T_i| \leq n / 2^i$ for all
$i\in[0:\ell]$,
we derive the statement of the lemma:
$$
\sum_{i=0}^{\ell}\frac{|T_i \cap \mathcal{U}|}{|T_i|} \geq  \sum_{i=0}^{\ell} \sum_{I \in T_i \cap \mathcal{U}} \frac{2^i}{n} \geq \frac{1}{n}\sum_{i=0}^{\ell}\sum_{I \in T_i \cap \mathcal{U}} \frac{|g(I)|}{2} \geq \frac{\eps}{8},
$$
where in the second inequality we used the fact that $|g(I)| = 2^{i}+1 \leq 2^{i+1}$, established above.
\end{proof}

\begin{proof}[Proof of~\Cref{lem:balanced_tree}] 
The tester always accepts when $f$ satisfies $\mathcal{P}$. Suppose now that $f$ is $\eps$-far from $\mathcal{P}$.
Let $T$ and $\mathcal{U}$ be as guaranteed in the statement of~\Cref{lem:structural_charac_unrepairability}. Consider a specific iteration of the loop in Algorithm \ref{alg:shifted_balanced_tree_tester_local_properties}. For a specific value of $i \in {[0:\ell]}$, the probability that the tester rejects is at least $\frac{|T_i \cap \mathcal{U}|}{|T_i|}$. Thus, the rejection probability for $i$ sampled uniformly from that range is at least
$$
\frac{1}{\ell+1} \sum_{i=0}^{\ell} \frac{|T_i \cap \mathcal{U}|}{|T_i|} \geq \frac{1}{\log (\eps n)} \cdot \frac{\eps}{8} = \frac{\eps}{8 \log (\eps n)},
$$
where the inequality follows
from~\Cref{lem:structural_charac_unrepairability}. That is, the rejection probability in a single iteration is at least $\alpha = \eps / 8 \log (\eps n)$, independently of other rounds. After $5/(2\alpha) = 20 \eps^{-1} \log (\eps n)$ iterations, the probability that the tester accepts is at most $(1-\alpha)^{5/(2\alpha)} < 1/7$, and it correctly rejects otherwise.
\end{proof}

\begin{lemma}[Marginal probability for pairs in structured tester]\label{lem:marginal_pair_balanced_tree_tester}
Let $\eps > 0$ and $n \geq 16/\eps$. Let $\ell$ be as in Algorithm \ref{alg:shifted_balanced_tree_tester_local_properties}. 
For $x, y \in [n]$, where $x<y,$
the probability that $x$ and $y$ are queried in Line~\ref{step:query-pair} of  Algorithm \ref{alg:shifted_balanced_tree_tester_local_properties} is at most $\frac{8}{n \log (\eps n)}$ if $y-x = 2^i$ for some $i \in [0:\ell]$, and equals zero otherwise. 
\end{lemma}
\begin{proof}
The definition of a query-pair $Q_{T,i}(\cdot)$ immediately implies the probability zero statement. 

Fix $i \in [0: \ell]$. Denote by $E_i$ the event that $i$ was sampled in Line~\ref{step:sample-level-interval} of Algorithm~\ref{alg:shifted_balanced_tree_tester_local_properties}. Note that $\Prob{}{E_i} = \frac{1}{\ell+1} < \frac{2}{\log (\eps n)}$.
For $x$ to be queried in a certain iteration (conditioned on $E_i$), the following two events must take place:
\begin{itemize}
\item \textit{Event A}: $x-a$ is an integer multiple of $2^i$, for $a$ as defined in Line 4 of Algorithm~\ref{alg:shifted_balanced_tree_tester_local_properties}.
\item \textit{Event B}: An interval sampled in Line~\ref{step:sample-level-interval} leads to $x$ being queried. 
\end{itemize}
We start by bounding
$\Prob{}{A | E_i}$. The number of multiples of $2^i$ in an interval of length $m = \lceil \eps n / 4 \rceil$ is at most $\frac{\eps n}{4 \cdot 2^i} + 1 \leq \frac{\eps n}{2^{i+1}}$. Since $a \in [m]$ is picked uniformly at random, 
$
\Prob{}{A | E_i} \leq \frac{\eps n}{2^{i+1}} \cdot \frac{4}{\eps n} \leq \frac{1}{2^{i-1}}.
$
We next bound $\Prob{}{B|A,E_i}$. There are at least $n/2^{i+1}$ intervals in $T_i$, out of which at most one leads to $x$ being queried. Thus
$\Prob{}{B|A,E_i} \leq \frac{2^{i+1}}{n}$. Combining all of our bounds, we have 
\[
\Prob{}{B \wedge A \wedge E_i} 
= \Prob{}{B|A,E_i} \cdot \Prob{}{A|E_i} \cdot \Prob{}{E_i}
\leq \frac{2^{i+1}}{n} \cdot \frac{1}{2^{i-1}} \cdot \frac{2}{\log (\eps n)} = \frac{8}{n \log (\eps n)}.
\qedhere
\]
\end{proof}

\Cref{lem:marginal_pair_balanced_tree_tester} allows us to show that Algorithm~\ref{alg:pair_tester_local_properties} can in a strong sense ``simulate'' the structured tester in Algorithm~\ref{alg:shifted_balanced_tree_tester_local_properties}. In the following lemma, we use such a simulation argument to prove that Algorithm~\ref{alg:pair_tester_local_properties} is indeed a valid pair tester (in an \emph{offline} setting) for $\mathcal{P}$.

\begin{proof}[Proof of~\Cref{lem:pair_tester_offline}]
By~\Cref{lem:balanced_tree}, Algorithm~\ref{alg:shifted_balanced_tree_tester_local_properties} is an $\eps$-tester for $\mathcal{P}$ with one-sided error in the offline model. We couple a single run of Algorithm~\ref{alg:shifted_balanced_tree_tester_local_properties} with parameters $n, \mathcal{P}, \eps$ with a single run of Algorithm~\ref{alg:pair_tester_local_properties} with the same parameters. Specifically, we couple (i) each iteration of the loop in Algorithm~\ref{alg:shifted_balanced_tree_tester_local_properties} with (ii) ten iterations of the loop in Algorithm~\ref{alg:pair_tester_local_properties}. 

For each pair $(x, y) \in D_d(n)$, if $d$ is not of the form $d=2^i$ for $i \in [0: \ell]$ (for $\ell$ as defined in both algorithms), then the probability that $x$ and $y$ are queried in both setups (i) and (ii) is zero. Suppose now that $d = 2^i$ for $i \in [0: \ell]$. By~\Cref{lem:marginal_pair_balanced_tree_tester}, the probability that $x$ and $y$ are queried in setting (i) is at most $\frac{8}{n \log (\eps n)}$. In Algorithm \ref{alg:pair_tester_local_properties}, the probability that $x$ and $y$ are queried in a single iteration of the loop is at least $\frac{1}{n} \cdot \frac{1}{\log (\eps n)} = \frac{1}{n \log (\eps n)}$. Each of the ten iterations in setting (ii) is independent, and so the probability that $(x,y)$ is \emph{not} queried in any of these iterations is bounded by $\left(1-\frac{1}{n \log (\eps n)}\right)^{10} < 1- \frac{8}{n \log (\eps n)}$. That is, $(x, y)$ is queried in at least one iteration with probability at least $\frac{8}{n \log (\eps n)}$. The lemma follows since loop iterations in both algorithms are independent.
\end{proof}

\subsection{Impossibility Result: No Tester for \texorpdfstring{$\omega(\eps)$}{epsilon} Erasure Rate}
\label{sec:impossibility_sortedness}
\begin{proof}[Proof of \Cref{thm:sortedness-budget-managing}, \Cref{item:impossibility-sortedness}  (impossibility)]
Let $\eps \in \big(0, \frac 13\big]$. Let $n \in \N$ be even. Using the same notation as in \cite{KalemajRV22}, we define the following distributions $\mathcal{D}^+$ and $\mathcal{D}^-$ over functions $f : [n] \to [n]$. This is the construction proposed and analyzed in \cite{KalemajRV22},
but we replaced $1/3$ with~$\eps.$
\paragraph{$\mathcal{D}^+$ distribution:} 
independently for all $i \in [n/2]$:
\begin{itemize}
    \item $f(2i-1) = f(2i) = 2i-1$ with probability $\eps$. (In this case, we call $(2i-1,2i)$ a ``low pair''.) 
    \item $f(2i-1)= f(2i) = 2i$ with probability $\eps$. (Here we call $(2i-1, 2i)$ a ``high pair.'')
    \item $f(2i-1) = 2i-1$ and $f(2i) = 2i$ with probability $1-2\eps$. (``Increasing pair.'')
\end{itemize}

\paragraph{$\mathcal{D}^-$ distribution:} 
independently for all $i \in [n/2]$:
\begin{itemize}
    \item $f(2i-1) = 2i$ and $f(2i) = 2i-1$  with probability $\eps$. (``Violating pair.'')
    \item $f(2i-1) = 2i-1$ and $f(2i) = 2i$ with probability $1-\eps$. (Increasing pair.)
\end{itemize}
All functions $f \in \mathcal{D}^+$ are monotone, whereas the expected distance of a function $f \in \mathcal{D}^-$ from monotonicity is $\eps$. 
It is straightforward to verify using Chernoff bound that the distance of $f \in \mathcal{D}^{-}$ to monotonicity is bigger than $\eps/2$ with probability $99/100$ as long as, say, $n > 20/\eps$; see the discussion of Yao's minimax principle after~\Cref{claim:low-degree-lb} and \cite[Corollary 9.4]{KalemajRV22} for more details.

Consider the following budget-managing adversary in the erasure model.\footnote{Our description of the adversarial strategy assumes that the adversary knows whether $f$ was generated from $\mathcal{D}^+$ or $\mathcal{D}^-$. However, this assumption is not really needed: the only function that can be generated with nonzero probability from both of these distributions is the identity function $f(x)=x$ for all $x$, and since the generation probability is exponentially small in $n$, this does not materially affect the analysis.}  
\begin{itemize}
    \item When $f \in \mathcal{D}^-$: If the tester queries an element of a violating pair, then the adversary erases the other element of this pair (i.e., the adversary erases $2i$ if $2i-1$ was queried,  and vice versa).
    If the tester queries an element of an increasing pair, then the adversary erases the other element with probability $\eps / (1-\eps)$.
    \item When $f \in \mathcal{D}^+$: If the tester queries an element of a low pair or a high pair, the adversary erases the other element. Otherwise (the increasing pair case), the adversary does not erase.
\end{itemize}

\paragraph{Indistinguishability {of} 
$\mathcal{D}^+$ and $\mathcal{D}^-$ under adversarial erasures.}
We claim the tester cannot distinguish 
these two distributions given the above adversarial strategy. Indeed, first observe that if an element $x$ was queried and its paired element was not erased, then they necessarily constitute an increasing pair. Thus, the only information that the adversary can extract by querying the element $x'$ paired to $x$ is the probability that $x'$ is subsequently erased after $x$ is queried, conditioned on the value of $f(x)$.

Suppose that, say, $x=2i-1$ (the case $x = 2i$ is symmetric). First, consider distribution $\mathcal{D}^+$.
If $f(x) = x = 2i-1$, then we are either at the ``low pair'' or ``increasing pair'' case, and the other element $x'$ is erased only in the low pair case. This happens with probability $$\frac{\eps}{\eps + (1-2\eps)} = \frac{\eps}{1-\eps}.$$
If $f(x) = 2i = x+1$ in this case, then we are in the ``high pair'' category, and $x'$ is erased with probability $1$.
For the $\mathcal{D}^-$ distribution, by definition of the adversary, it erases $x'$ if $f(x) = x$ with probability $\eps / (1-\eps)$, and with probability $1$ if $f(x)=x+1$. Thus, the probabilities are equal for both distributions.

\paragraph{Budget analysis.} It remains to verify that for an adversarial erasure rate of $C \cdot \eps$ for some large enough absolute constant $C$, the adversary's budget suffices to run the above strategy at all times. We do so using a combination of Markov inequality (for the earlier parts of the process) and Chernoff inequality (for the later parts).

For any round $i\in\N$ of the process, let $X_i$ denote the random variable indicating whether the tester queried an element from a previous unqueried non-increasing pair of $f$. Note that the probability of each $X_i$, even conditioned on all previous outcomes, is at most $2 \eps$. Denote $Y_\ell = \sum_{i=1}^{\ell} X_i$ and note that $\E[Y_\ell] \leq 2\eps \cdot \ell$.
Let $\ell_0 = 2/(\eps \cdot C)$ be a crude upper bound for the time it takes the adversary to add an erasure to its budget. By time $\ell_0 m$ the adversary has a budget of at least $m$ erasures.
Note that if $Y_{\ell_0 m} < m$ for all $m\in\N$ then the adversary has sufficient erasure budget at all times, as time steps in between follow automatically.
Thus, in any execution where the adversary fails, there must exist some $m\in\N$ for which $Y_{\ell_0 m} \geq m$. We denote these bad events by $Z_m$ (for any $m\in\N$), and bound the probability of any of these events ever happening.
We use Chernoff inequality with the upper bound on the expectation:
\[
    \Prob{}{Z_m}
    = \Prob{}{Y_{\ell_0 m} \geq m}
    = \Prob{}{Y_{\ell_0 m} \geq (1+\delta) \cdot 2\eps \ell_0 m}
    \leq \exp \left(- \frac{\delta^2 \cdot 2\eps \ell_0 m}{2+\delta} \right),
\]
where we set $C = 60$, and $\delta = 14$, so that $(1+\delta)\cdot 2\eps\ell_0 = 1$.
This implies $\eps\ell_0 = 2/C = 1/30$, and $2\delta^2/(2+\delta) = 24.5$, leading to
\[
    \Prob{}{Z_m}
    \leq e^{- 49m/30}
    < (1/5)^m .
\]
Finally, by a union bound, the erasure budget is insufficient with probability at most
\[
    \Prob{}{\cup_{m\in\N} Z_m} 
    \sum_{m=1}^{\infty} \Prob{}{Z_m}
    < \sum_{m=1}^{\infty} \frac{1}{5^m}
    = \frac{1}{4} .
\]
Thus, with probability more than $3/4$, the adversary can maintain its strategy indefinitely.
                  \end{proof}

\section*{Acknowledgement}
We thank Shachar Lovett for referring us to \cite{Ben-EliezerHL12,KeevashS05}, which led to the result in \Cref{sec:low-degree-lb}.

 \short{\bibliographystyle{plainurl}}{\bibliographystyle{alpha}}
\bibliography{references} 

\newcommand{\etalchar}[1]{$^{#1}$}
\begin{thebibliography}{AKK{\etalchar{+}}05}

\bibitem[AJMR16]{AwasthiJMR16}
Pranjal Awasthi, Madhav Jha, Marco Molinaro, and Sofya Raskhodnikova.
\newblock Testing {Lipschitz} functions on hypergrid domains.
\newblock {\em Algorithmica}, 74(3):1055--1081, 2016.

\bibitem[AKK{\etalchar{+}}05]{AlonKKLR05}
Noga Alon, Tali Kaufman, Michael Krivelevich, Simon Litsyn, and Dana Ron.
\newblock Testing {Reed-Muller} codes.
\newblock {\em IEEE Transactions on Information Theory}, 51(11):4032--4039, 2005.

\bibitem[AS03]{AroraS03}
Sanjeev Arora and Madhu Sudan.
\newblock Improved low-degree testing and its applications.
\newblock {\em Combinatorica}, 23(3):365--426, 2003.

\bibitem[BCH{\etalchar{+}}96]{BellareCHKS96}
Mihir Bellare, Don Coppersmith, Johan H{\aa}stad, Marcos~A. Kiwi, and Madhu Sudan.
\newblock Linearity testing in characteristic two.
\newblock {\em {IEEE} Transactions on Information Theory}, 42(6):1781--1795, 1996.

\bibitem[Bec08]{Beck08}
J\'{o}zsef Beck.
\newblock {\em Combinatorial Games: Tic-Tac-Toe Theory}.
\newblock Cambridge: Cambridge University Press, 2008.

\bibitem[BEHL12]{Ben-EliezerHL12}
Ido Ben-Eliezer, Rani Hod, and Shachar Lovett.
\newblock Random low-degree polynomials are hard to approximate.
\newblock {\em Computational Complexity}, 21(1):63--81, 2012.

\bibitem[Bel18]{Belovs18}
Aleksandrs Belovs.
\newblock Adaptive lower bound for testing monotonicity on the line.
\newblock In {\em Proceedings of Approximation, Randomization, and Combinatorial Optimization. Algorithms and Techniques {(APPROX/RANDOM)}}, pages 31:1--31:10, 2018.

\bibitem[Ben19]{BE19}
Omri Ben{-}Eliezer.
\newblock Testing local properties of arrays.
\newblock In {\em 10th Innovations in Theoretical Computer Science Conference, {ITCS} 2019}, pages 11:1--11:20, 2019.

\bibitem[BFL91]{BabaiFL91}
L{\'{a}}szl{\'{o}} Babai, Lance Fortnow, and Carsten Lund.
\newblock Non-deterministic exponential time has two-prover interactive protocols.
\newblock {\em Computational Complexity}, 1:3--40, 1991.

\bibitem[BFLR20]{BenFLR20}
Omri Ben{-}Eliezer, Eldar Fischer, Amit Levi, and Ron~D. Rothblum.
\newblock Hard properties with (very) short {PCPP}s and their applications.
\newblock In {\em Proceedings, Innovations in Theoretical Computer Science (ITCS)}, pages 9:1--9:27, 2020.

\bibitem[BFLS91]{BabaiFLS91}
L{\'{a}}szl{\'{o}} Babai, Lance Fortnow, Leonid~A. Levin, and Mario Szegedy.
\newblock Checking computations in polylogarithmic time.
\newblock In {\em Proceedings, ACM Symposium on Theory of Computing (STOC)}, pages 21--31, 1991.

\bibitem[BGJ{\etalchar{+}}12]{BGJRW12}
Arnab Bhattacharyya, Elena Grigorescu, Kyomin Jung, Sofya Raskhodnikova, and David~P. Woodruff.
\newblock Transitive-closure spanners.
\newblock {\em SIAM Journal on Computing (SICOMP)}, 41(6):1380--1425, 2012.

\bibitem[BGLR93]{BellareGLR93}
Mihir Bellare, Shafi Goldwasser, Carsten Lund, and Alexander Russell.
\newblock Efficient probabilistically checkable proofs and applications to approximations.
\newblock In {\em Proceedings, ACM Symposium on Theory of Computing (STOC)}, pages 294--304, 1993.

\bibitem[BGS98]{BellareGS98}
Mihir Bellare, Oded Goldreich, and Madhu Sudan.
\newblock Free bits, {PCP}s, and nonapproximability-towards tight results.
\newblock {\em SIAM Journal on Computing (SICOMP)}, 27(3):804--915, 1998.

\bibitem[BKS{\etalchar{+}}10]{BKSSZ10}
Arnab Bhattacharyya, Swastik Kopparty, Grant Schoenebeck, Madhu Sudan, and David Zuckerman.
\newblock Optimal testing of {Reed-Muller} codes.
\newblock In {\em Proceedings, IEEE Symposium on Foundations of Computer Science {(FOCS)}}, pages 488--497, 2010.

\bibitem[BLR93]{BLR93}
Manuel Blum, Michael Luby, and Ronitt Rubinfeld.
\newblock Self-testing/correcting with applications to numerical problems.
\newblock {\em Journal of Computer and System Sciences}, 47(3):549--595, 1993.

\bibitem[BRY14]{BermanRY14}
Piotr Berman, Sofya Raskhodnikova, and Grigory Yaroslavtsev.
\newblock {$L_p$}-testing.
\newblock In {\em Proceedings, ACM Symposium on Theory of Computing (STOC)}, pages 164--173, 2014.

\bibitem[BS94]{BellareS94}
Mihir Bellare and Madhu Sudan.
\newblock Improved non-approximability results.
\newblock In {\em Proceedings, ACM Symposium on Theory of Computing (STOC)}, pages 184--193, 1994.

\bibitem[BSVW03]{Ben-SassonSVW03}
Eli Ben{-}Sasson, Madhu Sudan, Salil~P. Vadhan, and Avi Wigderson.
\newblock Randomness-efficient low degree tests and short {PCP}s via epsilon-biased sets.
\newblock In {\em Proceedings, ACM Symposium on Theory of Computing (STOC)}, pages 612--621, 2003.

\bibitem[CDJS17]{ChakrabartyDJS17}
Deeparnab Chakrabarty, Kashyap Dixit, Madhav Jha, and C.~Seshadhri.
\newblock Property testing on product distributions: Optimal testers for bounded derivative properties.
\newblock {\em ACM Transactions on Algorithms (TALG)}, 13(2):20:1--20:30, 2017.

\bibitem[CS13]{CS13}
Deeparnab Chakrabarty and C.~Seshadhri.
\newblock Optimal bounds for monotonicity and {Lipschitz} testing over hypercubes and hypergrids.
\newblock In {\em Proceedings, ACM Symposium on Theory of Computing (STOC)}, pages 419--428, 2013.

\bibitem[CS14]{ChSe14}
Deeparnab Chakrabarty and C.~Seshadhri.
\newblock An optimal lower bound for monotonicity testing over hypergrids.
\newblock {\em Theory of Computing}, 10:453--464, 2014.

\bibitem[DG13]{DinurG13}
Irit Dinur and Venkatesan Guruswami.
\newblock {PCP}s via low-degree long code and hardness for constrained hypergraph coloring.
\newblock In {\em Proceedings, IEEE Symposium on Foundations of Computer Science {(FOCS)}}, pages 340--349, 2013.

\bibitem[DGL{\etalchar{+}}99]{DGLRRS99}
Yevgeniy Dodis, Oded Goldreich, Eric Lehman, Sofya Raskhodnikova, Dana Ron, and Alex Samorodnitsky.
\newblock Improved testing algorithms for monotonicity.
\newblock In {\em Proceedings of Approximation, Randomization, and Combinatorial Optimization. Algorithms and Techniques {(APPROX/RANDOM)}}, pages 97--108, 1999.

\bibitem[DJRT13]{DixitJRT13}
Kashyap Dixit, Madhav Jha, Sofya Raskhodnikova, and Abhradeep Thakurta.
\newblock Testing the {Lipschitz} property over product distributions with applications to data privacy.
\newblock In {\em Theory of Cryptography Conference (TCC)}, pages 418--436, 2013.

\bibitem[DRTV18]{DixitRTV18}
Kashyap Dixit, Sofya Raskhodnikova, Abhradeep Thakurta, and Nithin Varma.
\newblock Erasure-resilient property testing.
\newblock {\em SIAM Journal on Computing (SICOMP)}, 47(2):295--329, 2018.

\bibitem[EKK{\etalchar{+}}00]{EKKRV00}
Funda Erg{\"{u}}n, Sampath Kannan, Ravi Kumar, Ronitt Rubinfeld, and Mahesh Viswanathan.
\newblock Spot-checkers.
\newblock {\em Journal of Computer and System Sciences}, 60(3):717--751, 2000.

\bibitem[FGL{\etalchar{+}}96]{FeigeGLSS96}
Uriel Feige, Shafi Goldwasser, L{\'{a}}szl{\'{o}} Lov{\'{a}}sz, Shmuel Safra, and Mario Szegedy.
\newblock Interactive proofs and the hardness of approximating cliques.
\newblock {\em Journal of the ACM}, 43(2):268--292, 1996.

\bibitem[FS95]{FriedlS95}
Katalin Friedl and Madhu Sudan.
\newblock Some improvements to total degree tests.
\newblock In {\em Third Israel Symposium on Theory of Computing and Systems (ISTCS)}, pages 190--198, 1995.

\bibitem[GLR{\etalchar{+}}91]{GemmellLRSW91}
Peter Gemmell, Richard~J. Lipton, Ronitt Rubinfeld, Madhu Sudan, and Avi Wigderson.
\newblock Self-testing/correcting for polynomials and for approximate functions.
\newblock In {\em Proceedings, ACM Symposium on Theory of Computing (STOC)}, pages 32--42, 1991.

\bibitem[GR17]{goldreich2017learning}
Oded Goldreich and Dana Ron.
\newblock On learning and testing dynamic environments.
\newblock {\em Journal of the ACM (JACM)}, 64(3):1--90, 2017.

\bibitem[HKSS14]{HefetzKSS14}
Dan Hefetz, Michael Krivelevich, Miloš Stojakovi\'{c}, and Tibor Szab\'{o}.
\newblock {\em Positional Games}.
\newblock Oberwolfach Seminars. Vol. 44. Basel: Birkh\"{a}user Verlag GmbH, 2014.

\bibitem[HSS13]{HaramatySS13}
Elad Haramaty, Amir Shpilka, and Madhu Sudan.
\newblock Optimal testing of multivariate polynomials over small prime fields.
\newblock {\em SIAM Journal on Computing (SICOMP)}, 42(2):536--562, 2013.

\bibitem[HW03]{HastadW03}
Johan H{\aa}stad and Avi Wigderson.
\newblock Simple analysis of graph tests for linearity and {PCP}.
\newblock {\em Random Structures and Algorithms}, 22(2):139--160, 2003.

\bibitem[JPRZ09]{JutlaPRZ09}
Charanjit~S. Jutla, Anindya~C. Patthak, Atri Rudra, and David Zuckerman.
\newblock Testing low-degree polynomials over prime fields.
\newblock {\em Random Structures and Algorithms}, 35(2):163--193, 2009.

\bibitem[JR13]{JhaR13}
Madhav Jha and Sofya Raskhodnikova.
\newblock Testing and reconstruction of {Lipschitz} functions with applications to data privacy.
\newblock {\em SIAM Journal on Computing (SICOMP)}, 42(2):700--731, 2013.

\bibitem[KLX10]{KaufmanLX10}
Tali Kaufman, Simon Litsyn, and Ning Xie.
\newblock Breaking the epsilon-soundness bound of the linearity test over {GF(2)}.
\newblock {\em SIAM Journal on Computing (SICOMP)}, 39(5):1988--2003, 2010.

\bibitem[KM22]{kaufman2022improved}
Tali Kaufman and Dor Minzer.
\newblock Improved optimal testing results from global hypercontractivity.
\newblock In {\em 2022 IEEE 63rd Annual Symposium on Foundations of Computer Science (FOCS)}, pages 98--109. IEEE, 2022.

\bibitem[KR06]{KaufmanR06}
Tali Kaufman and Dana Ron.
\newblock Testing polynomials over general fields.
\newblock {\em SIAM Journal on Computing (SICOMP)}, 36(3):779--802, 2006.

\bibitem[KRV23]{KalemajRV22}
Iden Kalemaj, Sofya Raskhodnikova, and Nithin Varma.
\newblock Sublinear-time computation in the presence of online erasures.
\newblock {\em Theory of Computing}, 19(1):1--48, 2023.

\bibitem[KS05]{KeevashS05}
Peter Keevash and Benny Sudakov.
\newblock Set systems with restricted cross-intersections and the minimum rank ofinclusion matrices.
\newblock {\em SIAM Journal on Discrete Mathematics}, 18(4):713--727, 2005.

\bibitem[KS08]{KS08}
Tali Kaufman and Madhu Sudan.
\newblock Algebraic property testing: the role of invariance.
\newblock In {\em Proceedings of the fortieth annual ACM symposium on Theory of computing}, pages 403--412, 2008.

\bibitem[LPRV21]{LPRV21}
Amit Levi, Ramesh Krishnan~S. Pallavoor, Sofya Raskhodnikova, and Nithin Varma.
\newblock Erasure-resilient sublinear-time graph algorithms.
\newblock In {\em Proceedings, Innovations in Theoretical Computer Science (ITCS)}, pages 80:1--80:20, 2021.

\bibitem[Mos17]{Moshkovitz17}
Dana Moshkovitz.
\newblock Low-degree test with polynomially small error.
\newblock {\em Computational Complexity}, 26(3):531--582, 2017.

\bibitem[MR08]{MoshkovitzR08}
Dana Moshkovitz and Ran Raz.
\newblock Sub-constant error low degree test of almost-linear size.
\newblock {\em SIAM Journal on Computing (SICOMP)}, 38(1):140--180, 2008.

\bibitem[MZ23]{MinZ}
Dor Minzer and Kai Zheng.
\newblock Adversarial low degree testing.
\newblock {\em arXiv}, 2308.15441, 2023.

\bibitem[NR21]{NakarR21}
Yonatan Nakar and Dana Ron.
\newblock {Testing Dynamic Environments: Back to Basics}.
\newblock In Nikhil Bansal, Emanuela Merelli, and James Worrell, editors, {\em 48th International Colloquium on Automata, Languages, and Programming (ICALP 2021)}, volume 198 of {\em Leibniz International Proceedings in Informatics (LIPIcs)}, pages 98:1--98:20, Dagstuhl, Germany, 2021. Schloss Dagstuhl -- Leibniz-Zentrum f{\"u}r Informatik.

\bibitem[NV21]{NV20}
Ilan Newman and Nithin Varma.
\newblock New sublinear algorithms and lower bounds for {LIS} estimation.
\newblock In {\em Proceedings, International Colloquium on Automata, Languages and Programming (ICALP)}, pages 100:1--100:20, 2021.

\bibitem[O'D14]{OD1}
Ryan O'Donnell.
\newblock {\em Analysis of Boolean Functions}.
\newblock Cambridge University Press, 2014.

\bibitem[PRV18]{PRV18}
Ramesh Krishnan~S. Pallavoor, Sofya Raskhodnikova, and Nithin Varma.
\newblock Parameterized property testing of functions.
\newblock {\em {ACM} Transactions on Computation Theory}, 9(4):17:1--17:19, 2018.

\bibitem[PRW22]{PallavoorRW22}
Ramesh Krishnan~S. Pallavoor, Sofya Raskhodnikova, and Erik Waingarten.
\newblock Approximating the distance to monotonicity of boolean functions.
\newblock {\em Random Struct. Algorithms}, 60(2):233--260, 2022.

\bibitem[Ras99]{Ras99}
Sofya Raskhodnikova.
\newblock Monotonicity testing.
\newblock {\em Masters Thesis, MIT}, 1999.

\bibitem[Ras16]{Enc1}
Sofya Raskhodnikova.
\newblock Testing if an array is sorted.
\newblock {\em Encyclopedia of Algorithms}, pages 2219--2222, 2016.

\bibitem[RR16]{RasR16}
Sofya Raskhodnikova and Ronitt Rubinfeld.
\newblock Linearity testing/testing {H}adamard codes.
\newblock In {\em Encyclopedia of Algorithms}, pages 1107--1110. Springer, 2016.

\bibitem[RRV21]{RRV19}
Sofya Raskhodnikova, Noga Ron{-}Zewi, and Nithin Varma.
\newblock Erasures versus errors in local decoding and property testing.
\newblock {\em Random Structures and Algorithms}, 59(4):640--670, 2021.

\bibitem[RS96]{RubinfeldS96}
Ronitt Rubinfeld and Madhu Sudan.
\newblock Robust characterizations of polynomials with applications to program testing.
\newblock {\em SIAM Journal on Computing (SICOMP)}, 25(2):252--271, 1996.

\bibitem[RS97]{RazS97}
Ran Raz and Shmuel Safra.
\newblock A sub-constant error-probability low-degree test, and a sub-constant error-probability {PCP} characterization of {NP}.
\newblock In {\em Proceedings, ACM Symposium on Theory of Computing (STOC)}, pages 475--484, 1997.

\bibitem[RS13]{Ron-ZewiS13}
Noga Ron{-}Zewi and Madhu Sudan.
\newblock A new upper bound on the query complexity of testing generalized {Reed-Muller} codes.
\newblock {\em Theory of Computing}, 9:783--807, 2013.

\bibitem[RV18]{RV18}
Sofya Raskhodnikova and Nithin Varma.
\newblock Brief announcement: Erasure-resilience versus tolerance to errors.
\newblock In {\em Proceedings, International Colloquium on Automata, Languages and Programming (ICALP)}, pages 111:1--111:3, 2018.

\bibitem[Sam07]{Samorodnitsky07}
Alex Samorodnitsky.
\newblock Low-degree tests at large distances.
\newblock In {\em Proceedings, ACM Symposium on Theory of Computing (STOC)}, pages 506--515, 2007.

\bibitem[ST98]{SudanT98}
Madhu Sudan and Luca Trevisan.
\newblock Probabilistically checkable proofs with low amortized query complexity.
\newblock In {\em Proceedings, IEEE Symposium on Foundations of Computer Science {(FOCS)}}, pages 18--27, 1998.

\bibitem[ST00]{SamorodnitskyT00}
Alex Samorodnitsky and Luca Trevisan.
\newblock A {PCP} characterization of {NP} with optimal amortized query complexity.
\newblock In {\em Proceedings, ACM Symposium on Theory of Computing (STOC)}, pages 191--199, 2000.

\bibitem[ST09]{SamorodnitskyT09}
Alex Samorodnitsky and Luca Trevisan.
\newblock Gowers uniformity, influence of variables, and {PCP}s.
\newblock {\em SIAM Journal on Computing (SICOMP)}, 39(1):323--360, 2009.

\bibitem[SW06]{ShpilkaW06}
Amir Shpilka and Avi Wigderson.
\newblock Derandomizing homomorphism testing in general groups.
\newblock {\em SIAM Journal on Computing (SICOMP)}, 36(4):1215--1230, 2006.

\bibitem[Tre98]{Trevisan98}
Luca Trevisan.
\newblock Recycling queries in {PCP}s and in linearity tests (extended abstract).
\newblock In {\em Proceedings, ACM Symposium on Theory of Computing (STOC)}, pages 299--308, 1998.

\end{thebibliography}

\pagebreak
\appendix
\section{Observation about Low-Degree Testing with Large Batches}\label{sec:low-degree-testing-with-large-batches}
This section is devoted to the proof of the following observation on online testing of property $\cP_d$ (being a polynomial of degree at most $d$) with batches of size $b=2^{d+1}.$
\begin{observation}
    \label{obs:d_deg:one_cube}
    Fix $d,n\in\N$, $\eps \in (0, 1/2)$, and $t$ such that $t \leq \min\set{\eps^2 2^{2d}, \eta} \cdot 2^{n-d}/10^5$ for some constant $\eta$, and $d < n/3$. Then there exists an $\eps$-tester for property $\cP_d$ of functions $f:\{0,1\}^n\to \{0,1\}$ that works in the presence of a $(2^{d+1},t)$-online erasure budget-managing adversary and makes
    $O\left(\frac{1}{\eps} + 2^d\right)$ queries.
\end{observation}

\begin{proof}
    With batch size $b = 2^{d+1}$, one could use a single batch to run the standard degree-$d$ test: choose a $(d+1)$-dim affine subspace uniformly at random, and reject if the restricted function is not of degree $d$ -- that is, the XOR of all returned values is $1$ (and in particular, none of them is erased).
    By \cite[Theorem~1]{BKSSZ10}, there exists a constant $\zeta$, such that any $\eps$-far function is rejected by the test with probability at least $s:= \min\set{2^{d+1} \eps, \zeta}$, while any degree-$d$ function is always accepted. We next consider $r = 10/s$ iterations of this atomic test, an analyze the probability of seeing an erasure.

    For completeness, note that erasures cannot cause the tester to reject, so a degree-$d$ function is always accepted. 

    For soundness, fix an $\eps$-far function $f$, and consider the iteration $i+1\in [r]$. At most $i\cdot t$ points were erased so far. Each queried point has a marginal distribution of a uniformly random point, with probability at most $i \cdot t / 2^{n}$ to be a previously erased point.
    Using a union bound over all $2^{d+1}$ points, the probability of seeing \emph{any} erasure during iteration $i$ is at most
    \[
        \frac{i t 2^{d+1}}{2^n}
        \leq \frac{10 t 2^{d+1}}{s\cdot 2^n}
        \leq \frac{s}{100} ,
    \]
    where the last inequality is due to the premise on $t$ (choosing $\eta = \zeta^2$ which is constant).
    
    Overall, the probability of the test choosing a pattern that would (wrongfully) pass the test is at most $1-s$, and the probability of seeing an erasure is at most $s/100$. By a union bound, the probability to accept $f$ at any iteration is at most $1-s/2$.
    Using $r = 10/s$ iterations, the probability to accept an $\eps$-far function is at most $1/3$.

    The query complexity is $2^{d+1} \cdot r = O\left(\max \set{\frac{1}{\eps}, 2^d}\right)$.

    Note that the probability of seeing \emph{any} erasure is at most $(s/100) \cdot (10/s) = 1/10$, which allows this tester to deal with corruptions as well, by \cite[Lemma 1.8]{KalemajRV22}.
\end{proof}

\section{Proof of Generalized Witness - \autoref{claim:low_deg_characterization}}\label{sec:proof_of_gen_witness}

\begin{proof}[Proof of \Cref{claim:low_deg_characterization}]
 For every $J\subseteq[n],$ we denote by $g_J\colon \F_2^n\to\F_2$ the function which consists only of one monomial corresponding to the variables in the set $J$, i.e., $g_J(x_1,\dots,x_n)=\prod_{j\in J}x_j$.  
Note that the class of degree $d$ functions, $\mathcal{P}_d$, is spanned by the degree at most $d$ monomials. $$\mathcal{P}_d=sp\{g_J\colon J\subseteq[n],|J|\leq d\}.$$ 
Let us first expand the testing pattern for the monomial functions.
Fix $J\subseteq [n]$  and $M\in \mathcal{M}_{m \times n}(\F_2)$, let $r=|J|$, then 
 \begin{align*}
          T_{X,g_J}(M)&=T_{X,g_J}(M_1,\dots,M_m)\\
          &=\sum_{i\in[\ell]} g_J((XM)_i)\\
          &=\sum_{i\in[\ell]} \prod_{j\in[r]}(XM)_{i,J_j}\\
          &=\sum_{i\in[\ell]} \prod_{j\in[r]}\sum_{k\in[m]} X_{i,k}M_{k,J_j}\\
          &=\sum_{i\in[\ell]} \bigg(\sum_{k\in[m]} X_{i,k}M_{k,J_1}\bigg)\bigg(\sum_{k\in[m]} X_{i,k}M_{k,J_2}\bigg)\dots \bigg(\sum_{k\in[m]} X_{i,k}M_{k,J_{r}}\bigg)\\
          &=\sum_{i\in[\ell]} \bigg(\sum_{k_1\in[m]} X_{i,k_1}M_{k_1,J_1}\bigg)\bigg(\sum_{k_2\in[m]} X_{i,k_2}M_{k_2,J_2}\bigg)\dots \bigg(\sum_{k_{r}\in[m]} X_{i,k_{r}}M_{k_{r},J_{r}}\bigg)\\
          &=\sum_{i\in[\ell]} \sum_{k_1\in[m]}\sum_{k_2\in[m]}\dots \sum_{k_{r}\in[m]} \bigg(X_{i,k_1}X_{i,k_2}\dots X_{i,k_{r}} \bigg)\bigg( M_{k_1,J_1}M_{k_2,J_2}\dots  M_{k_{r},J_{r}}\bigg)\\
          &=\sum_{i\in[\ell]} \sum_{k_1\in[m]}\sum_{k_2\in[m]}\dots \sum_{k_{r}\in[m]} \prod_{a\in[r]}X_{i,k_a} \prod_{b\in[r]}M_{k_b,J_b}\\
          &=\sum_{k_1\in[m]}\sum_{k_2\in[m]}\dots \sum_{k_{r}\in[m]}  \prod_{b\in[r]}M_{k_b,J_b} \sum_{i\in[\ell]}\prod_{a\in[r]}X_{i,k_a}.
      \end{align*}

      Denote $$\varphi_S(X)=\sum_{i\in[\ell]} \prod_{j\in S} X_{ij}.$$

      Then we have 
      \begin{equation}\label{eq:monomial_testing_expansion}
          T_{X,g_J}(M)=\sum_{k_1\in[m]}\sum_{k_2\in[m]}\dots \sum_{k_{r}\in[m]}  \prod_{b\in[r]}M_{k_b,J_b} \varphi_{\set{k_1,k_2,\dots, k_r}}(X).
      \end{equation}

      Let $X$ be a complete pattern for $\mathcal{P}_d$ and fix $S\subseteq [m]$ with $r=|S|\leq d$ we want to show that $\varphi_S(X)=0$. Let $M'\in \set{0,1}^{m\times n}$ be defined as $M'_{i,j}=\mathrm{1}_{i=j}.$ Since $g_S\in\mathcal{P}_d$ by the completeness of $X$ we have $T_{X,g_S}(M')=0$ which by \eqref{eq:monomial_testing_expansion} equals 
      $$0=T_{X,g_S}(M')=\sum_{k_1\in[m]}\sum_{k_2\in[m]}\dots \sum_{k_{r}\in[m]}  \prod_{b\in[r]}M'_{k_b,S_b} \varphi_{\set{k_1,k_2,\dots, k_r}}(X)=\varphi_S(X),$$
where the last equality follows from the definition of $M'$. 

If, in addition, $X$ is minimally sound for $\mathcal{P}_d$ as in definition \ref{def:pattern_completeness_minimaly_sound}, let $J=\set{1,2,\dots,d+1}$ and note that $g_J\not\in\cP_d$ then there exists $M'\in\mathcal{M}_{m \times n}(\F_2)$ such that $T_{X,g_J}(M')=1$, by \eqref{eq:monomial_testing_expansion} there must exist a setting of $(k'_1,k'_2,\dots,k'_{d+1})\in[m]^{d+1} $ such that $$\prod_{b\in[d+1]}M'_{k'_b,J_b} \varphi_{\set{k'_1,k'_2,\dots, k'_{d+1}}}(X)=1,$$
which implies $\varphi_{\set{k'_1,k'_2,\dots, k'_{d+1}}}(X)=1$. Now, by condition~\eqref{itm:completeness_property}  $k'_1,k'_2,\dots,k'_{d+1}$ must be distinct. Thus we prove condition~\eqref{itm:soundness_property} for $S=\set{k'_1,k'_2,\dots, k'_{d+1}}$ of size $d+1$.  

For the other direction, assume  $X$ satisfies condition~\eqref{itm:completeness_property}  we will show that  $X$ is complete for $\cP_d$. Since $\cP_d$ is spanned by the set of monomials of size at most $d$, it is enough to show that $X$ is complete for all monomials $g_J$ where $J\subseteq [n]$, $|J|\leq d$. 
      Fix $J\subseteq [n]$ such that $r=|J|\leq d$ and $M \in \set{0,1}^{m\times n}$. We apply property \eqref{itm:completeness_property} to show that every summand in \eqref{eq:monomial_testing_expansion} is $0$, where we note that the sets $k_1,k_2,\dots, k_r$ contains at most $r\leq d$ elements, thus $T_{X,g_J}(M)=0$.

      If, in addition, $X$ satisfies condition~\eqref{itm:soundness_property} then let $S^\star=\set{s_1,s_2,\dots,s_{d+1}}$ be such that $\varphi_{S^\star}(X)=1$, fix $f\not\in\cP_d$ and we can write $f=\sum_{J\subseteq[n]}\alpha_J g_J$, let $r^\star>d$ be the minimal size of monomial with $|J|>d$ and  $\alpha_J=1$ and let $J^\star=\set{j_1,j_2,\dots, j_{r^\star}}$ be such a monomial. Then for $M^\star\in\set{0,1}^{m\times n}$ defined as following for  each  $j_i\in J^\star$, if $i\leq d+1$ then $M^\star_{s_i,j_i}=1$, if $i>d+1$ then  $M^\star_{s_{d+1},j_i}=1$  and $0$ in all other entries.
      
    \begin{align*}
        T_{X,f}(M^\star)& =\sum_{i\in[\ell]}f((XM^\star)_i)\\ 
        & =\sum_{i\in[\ell]}\sum_{J\subseteq[n]}\alpha_Jg_J((XM^\star)_i)\\
        &=\sum_{J\subseteq[n]}\alpha_J T_{X,g_J}(M^\star)\\
        &=\sum_{J\subseteq[n]\colon |J|\geq r^\star}\alpha_J T_{X,g_J}(M^\star).
    \end{align*}
    Where the last equality is by the completeness and minimality of $r^\star$.
    Now for $J^\star$ we have $T_{X,g_{J^\star}}(M^\star)=\varphi_{S^\star}(X)=1$ and for every other $J\neq J^\star$ we have $|J|\geq r^\star$ and therefore there is at least one $j\in J\setminus J^\star$ for which the $j$'th column of $M^\star$ is all zero and thus by \eqref{eq:monomial_testing_expansion}
 we have $T_{X,g_J}(M^\star)=0$.
 This implies 
 $$T_{X,f}(M^\star)=T_{X,g_{J^\star}}(M^\star)=1.$$

\end{proof}

\end{document}